%% file: arxiv.tex
\documentclass[11pt]{article}
\renewcommand{\baselinestretch}{1.4}
\usepackage{amsmath,epsfig, amssymb, amsthm, latexsym,color, natbib, euscript}
\usepackage{graphicx}
\usepackage{fancyhdr}
\usepackage{algpseudocode}
\usepackage{float}
\makeatletter

\input{QYalias}

\newcommand{\Date}[1]{\def\@Date{#1}}
\def\today{\number\day~\ifcase\month\or
	January\or February\or March\or April\or May\or June\or
	July\or August\or September\or October\or November\or December\fi~\number\year}

\def \b1{{\bf 1}}

\begin{document}
	\title{\bf Blind Source Separation over Space}
	\author{
		Bo Zhang\\
		Department of Statistics and Finance, International Institute of Finance\\
		School of Management, University of Science and Technology of China, Hefei, China\\
		zhangbo890301@outlook.com\\[1ex]
		Sixing Hao \quad and \quad Qiwei Yao\\
		Department of Statistics,  London School of Economics,  London, WC2A 2AE, UK\\
		s.hao3@lse.ac.uk \qquad q.yao@lse.ac.uk
	}
	\maketitle
	
\begin{abstract}
	We propose a new estimation method for the blind source separation model
	of \cite{bgnrv20}. The new estimation is based on an eigenanalysis of a positive
	definite matrix defined in terms of multiple normalized spatial
	local covariance matrices, and, therefore, can handle moderately high-dimensional
	random fields. The consistency of the estimated mixing matrix is established
	with explicit error rates even when the eigen-gap decays to zero slowly.
	The proposed method is illustrated via both simulation
	and a real data example.
\end{abstract}
	
	\noindent
	{\sl Some key words}: Eigen-analysis; Eigen-gap; High-dimensional random field;
	Mixing matrix; Spatial local covariance matrix.

	\section{Introduction}
	Blind source separation is an effective way to reduce the complexity in modelling
	$p$-variant spatial data \citep{nofr15, bgnrv20}. The approach can be viewed as
	a version of independent component analysis \citep{hko01} for
	multivariate spatial random fields.
	Though only the second moment properties are concerned, the challenge is to
	decorrelate $p$  spatial random fields at the same location
	as well as across different locations.
	Note that the standard principal component analysis does not
	capture spatial correlations, as it only diagonalizes the covariance
	matrix (at the same location).
	\cite{nofr15} introduced a so-called local covariance matrix to represent
	the dependence across different locations. Furthermore, it proposed to estimate
	the mixing matrix in the blind source separation decomposition based on
	a generalized eigenanalysis, which
	can be viewed as an extension of the
	principal component analysis as it diagonalizes a local covariance matrix
	in addition to the standard covariance matrix.
	To overcome the drawback of using the information from only one
	local covariance
	matrix, \cite{bgnrv20} proposed to use multiple local covariance matrices in
	the estimation. The method of \cite{bgnrv20} has a clear advantage in
	incorporating the spatial dependence information
	over different ranges.
	However, the estimation is based on a nonlinear optimization with $p^2$
	parameters. Hence it is compute-intensive and cannot cope with very large $p$.
	
	Inspired by \cite{bgnrv20}, we propose a new method also based on multiple (normalized)
	local covariance matrices for estimating the
	mixing matrix. Different from \cite{bgnrv20}, the new method is
	computationally efficient as it boils down to an
	eigenanalysis of a positive definite matrix
	which is a matrix function of multiple normalized spatial local covariance matrices.
	Therefore it can handle the cases with the dimension of
	random fields in the order of a few thousands on an ordinary personal
	computer. While the basic idea resembles that of \cite{cgy18} which
	dealt with multiple time series, the spatial random fields concerned
	are sampled irregularly and non-unilaterally, and the spatial correlations spread in all
	directions. Furthermore, we incorporate the pre-whitening in our search for the mixing matrix. This implies estimating the covariance matrix of
	the process, which is assumed to be an identity matrix in  \cite{cgy18}.
	The normalized spatial local covariance matrix is a modified version
	of the spatial local covariance matrix in \cite{nofr15}, and is introduced to facilitate the effect of the pre-whitening.
	All these entail completely different theoretical exploration;
	leading to the asymptotic results under the similar setting of \cite{bgnrv20}
	but allowing the dimension of the random field to diverge together
	with the number of the observed locations, which is assumed to be fixed
	in \cite{bgnrv20}.
	
	Another new contribution of the paper concerns the eigen-gap in the
	eigenanalysis for estimating the mixing matrix. In order to identify
	a consistent estimator for the mixing matrix, the standard condition
	is to assume that the minimum pairwise absolute difference among the eigenvalues
	remains positive.  See Assumptions 8 and 9 of \cite{bgnrv20}.
	The similar conditions have been imposed in the literature
	in order to identify factor loading spaces in factor models
	\citep{ly12}.
	However this condition is invalid under the setting concerned in
	this paper when the dimension of random field $p$ diverges to
	infinity, as the maximum order of the eigen-gap is $p^{-1}$.
	We show that the identification of the mixing matrix is still
	possible when $p \to \infty$ at the rate $p=o(n^{1/3})$.  See Theorem
	\ref{t3} and Remark 1 in Section \ref{sec3}.
	
	The rest of the paper is organised as follows. We present the spatial blind source
	separation model and the new estimation method in Section \ref{sec2}. The asymptotic
	properties are developed in Section \ref{sec3}. Numerical illustration with both
	simulated data and a real data set is presented in Section \ref{sec4}. All the technical
	proofs are given in Section \ref{sec5}.

	\section{Setting and Methodology} \label{sec2}
	
	\subsection{Model} \label{sec21}

We adopt the spatial blind source separation model of \cite{bgnrv20}.
More precisely, let
$X(s) = \{ X_1(s), \cdots, X_p(s) \}^\top $ be a $p$-variate random field
defined on $s \in \calS \subset \RR^d$, and  $X(s)$ admits the representation
\begin{equation} \label{b1}
	X(s) = \Omega Z(s) \equiv \Omega \{ Z_1(s), \cdots, Z_p(s) \}^\top,
\end{equation}
where $ Z_1(s), \cdots, Z_p(s) $ are $p$ independent latent random fields,
and $\Omega$ is a $p\times p$ invertible constant matrix and
is called the mixing matrix. Furthermore, \cite{bgnrv20} assumes that for any $s, u \in \calS$,
\begin{equation} \label{b2}
	E Z(s) = \mu_0, \quad \var\{ Z(s) \} = I_p, \quad \cov\{Z(s), Z(u) \} = H(s - u),
\end{equation}
where  $\mu_0$ is an unknown constant vector, $I_p$ denotes the $p\times p$ identity matrix, $H(\cdot )$ is a $p\times p$ diagonal
matrix
\[
H(s - u) = \diag\{ K_1(s-u), \cdots, K_p(s-u) \},
\]
i.e. $\cov\{ Z_i(s), Z_j(u) \} = K_i(s-u)$ if $i=j$, and 0 otherwise. Let
$\mu= \Omega \mu_0$.
Under (\ref{b1}) and (\ref{b2}), $X(\cdot)$ is a weakly stationary process as
\begin{equation} \label{b3}
	EX(s) =\mu, \quad \var\{X(s)\} = \Omega \Omega^\top, \quad  \cov\{X(s), X(u) \}
	= \Omega H(s - u) \Omega^\top.
\end{equation}

\subsection{Estimation method}\label{sec22}

Let $X(s_1), \cdots, X(s_n)$ be available observations.
Put
\[
\tilde X(s_i) =X(s_i) - {1 \over n} \sum_{j=1}^n X(s_j), \quad
\tilde Z(s_i) = Z(s_i) - {1 \over n} \sum_{j=1}^n Z(s_j),
\quad i=1, \cdots, n.
\]
Then the spatial local covariance matrix of \cite{nofr15}
is defined as
\begin{equation} \label{b5}
	\tilde M(f) = {1 \over n} \sum_{i,j=1}^n f(s_i - s_j) \tilde X(s_i) \tilde X(s_j)^\top,
\end{equation}
where $f(\cdot)$ is a kernel function such as $f(s) = 1(h_1 \le \|s\|\le h_2)$
for some constants $0\le h_1 < h_2 < \infty$, and $1(\cdot)$ denotes
the indicator function.
To recover the mixing matrix $\Omega$,
\cite{bgnrv20} propose to estimate the
unmixing matrix (i.e. the inverse
of the mixing matrix) $\Gamma = \Omega^{-1} \equiv (\gamma_1, \cdots, \gamma_p)^\top$ by
\[
\hat \Gamma \; \in \;
\arg\max_{ \Gamma \tilde M(f_0) \Gamma^\top = I_p} \sum_{i=1}^k \sum_{j=1}^p
\{ \gamma_j^\top \tilde M(f_i) \gamma_j \}^2,
\]
where $f_0(s) = I(s=0)$, and $f_1, \cdots, f_k$ are appropriately specified kernels.
This is a nonlinear optimization problems with $p^2$ variables, which
\cite{bgnrv20} adopted the algorithm of \cite{c88} to solve. When $k=1$,
the objective function contains only one kernel function.
Then the above optimization can be solved
based on a generalized eigenanalysis; see \cite{nofr15} and \cite{bgnrv20}, though
the estimation based on a single kernel requires the prior knowledge
on which kernel to use for a given problem.

We now propose a new method to estimate the mixing matrix using multiple kernels but based on a single eigenanalysis.
To this end, we define,
for any given $k$ kernel function $f_1(\cdot), \cdots, f_k(\cdot)$,
\begin{align} \label{nfi}
	N &= E\Big[  {1 \over k}\sum_{h=1}^k \big\{ {1 \over n} \sum_{i,j=1}^n f_h(s_i - s_j)
	\tilde Z(s_i) \tilde Z(s_j)^\top \big\} \big\{ {1 \over n} \sum_{i,j=1}^n f_h(s_i - s_j)
	\tilde Z(s_i)  \tilde Z(s_j)^\top \big\}^\top\Big],\\ \nonumber
	W &= E\Big[ {1 \over k}\sum_{h=1}^k \big\{ {1 \over n} \sum_{i,j=1}^n f_h(s_i - s_j)
	\Sigma^{-1/2} \tilde X(s_i) \tilde X(s_j)^\top \big\} \Sigma^{-1} \\ \nonumber
	& \;\; \quad \times \big\{ {1 \over n} \sum_{i,j=1}^n f_h(s_i - s_j)
	\tilde X(s_i) \tilde X(s_j)^\top\Sigma^{-1/2} \big\}^\top \Big],
\end{align}
where $\Sigma = {\rm Var}\{X(s)\}= \Omega \Omega^\top$.
Then $N$ and $W$ are  $p\times p$ non-negative definite matrices. Furthermore,
$N$ is a diagonal matrix, as its $(i,j)$-th element, for $i\ne j$, is
\[
{1 \over n^2k} \sum_{h=1}^k \sum_{\ell =1}^p \sum_{i_1, i_2, j_1, j_2=1}^n
f_h(s_{i_1} - s_{j_1}) f_h(s_{i_2} - s_{j_2}) E\{ \tilde Z_i(s_{i_1}) \tilde Z_\ell(s_{j_1})
\tilde Z_j(s_{i_2}) \tilde Z_\ell (s_{j_2}) \} =0,
\]
which is guaranteed by the fact that
the components of $Z(\cdot)$ are the $p$ independent random fields.
Since $\Omega$ is a $p \times p$ full rank matrix, we can rewrite $\Omega=V_{\Omega}\Lambda_{\Omega}U_{\Omega}$, where $V_{\Omega}$ and $U_{\Omega}$ are two $p\times p$ orthogonal matrices, and $\Lambda_{\Omega}$ is a diagonal matrix. Then $\Sigma^{-1/2} =
V_{\Omega}\Lambda_{\Omega}^{-1}V_{\Omega}^\top$. Combining this and (\ref{b1}), we have
\begin{equation} \label{b9}
	W = V_{\Omega} U_{\Omega} N U_{\Omega}^\top V_{\Omega}^\top,
\end{equation}
i.e. the columns of $ U_W\equiv V_{\Omega} U_{\Omega}$ are the $p$ orthonormal eigenvectors of matrix $W$ with the diagonal elements of $N$ as the corresponding eigenvalues. As $\Sigma^{1/2} U_{W}=V_{\Omega}\Lambda_{\Omega}V_{\Omega}^\top V_{\Omega} U_{\Omega}=\Omega$, this paves the way to identifying
mixing matrix $\Omega$. We summarize the finding in
the proposition below.

\begin{proposition} \label{prop1}
	Under the condition (\ref{b2}), the  mixing matrix $\Omega$ defined
	in (\ref{b1}) is of the form $\Sigma^{1/2}U_W$, where the columns of $U_W$ are the $p$ orthonormal eigenvectors of matrix $W$.
	Moreover, those $p$ eigenvectors are identifiable, upto the sign changes, if the $p$ diagonal elements of
	$N$ are distinct from each other.
\end{proposition}

Note that the sign changes of any columns of $U_W$ will not change the
independence of the components of $Z(\cdot)$ in (\ref{b1}), as $Z(s) = U_W^\top\Sigma^{-1/2}
X(s)$.
By Proposition \ref{prop1}, we define an estimator for the mixing matrix as
\begin{equation} \label{newesti}
	\hat \Omega=\hat \Sigma^{1/2}\hat U_{W},
\end{equation}
where $\hat \Sigma=n^{-1} \sum_{1\le j \le n} \tilde X(s_j) \tilde X(s_j)^\top$, and
the columns of $\hat U_{W}$ are the $p$ orthonormal eigenvectors
of matrix
\begin{equation} \label{nfihatX}
	\hat W =  {1 \over k}\sum_{h=1}^k \hat M(f_h) \hat M(f_h)^\top.
\end{equation}
In the above expression, $\hat M(f_h)$ is a normalized
local covariance matrix defined as
\begin{equation} \label{b5scaleh}
	\hat M(f) = {1 \over n} \sum_{i,j=1}^n f(s_i - s_j)\hat \Sigma^{-1/2}  \tilde X(s_i) \tilde X(s_j)^\top \hat \Sigma^{-1/2}.
\end{equation}
In comparison to the local covariance matrix (\ref{b5}), we replace $X(\cdot)$ by its standardized version $\hat \Sigma^{-1/2} \tilde X(\cdot)$. This effectively pre-whitens the data in our search for the mixing matrix.

To end this section, we note that the proposed new method makes use of the normalized 4th moments of the observations while the methods of \cite{bgnrv20} and \cite{nofr15} only depend on the 2nd moments.

	\section{Asymptotic properties} \label{sec3}
	
	We consider the asymptotic behaviour of the estimator $\wh \Omega$ when $n \to \infty$ and $p$ either remaining fixed or $p = o(n)$. Since $\wh\Omega^{-1}X(s)=\wh\Omega^{-1} \Omega Z(s)$, we will focus on $\wh\Gamma_{\Omega}=\wh\Omega^{-1}\Omega$. We introduce some regularity conditions first.
	\begin{description}
		\item[A1.]  In model (\ref{b1}), $Z_1(\cdot), \cdots, Z_p(\cdot)$ are $p$
		independent and strictly stationary random fields on $R^d$, and condition (\ref{b2})  holds.
		Furthermore, $Z(\cdot)$ is sub-Gaussian in the sense
		that there exists a constant $C_0>0$ independent of $p$ for which
		\begin{equation} \label{subgaussianassu}
			\sup_{\beta \geq 1, 1\leq i \leq p}\beta^{-1/2}\{E|Z_i(s)|^\beta\}^{1/\beta} \leq C_0.
		\end{equation}
		Moreover, for any unit vector $(a_1,\cdots,a_n)^\top \in R^n$ and $1 \leq \ell \leq p$,
		$\sum_{i=1}^na_iZ_\ell(s_i)$ is sub-Gaussian.
		\item[A2.] There exist positive constants $\Delta, \alpha $
		and $A$ (independent of $n$ and $p$) such that
		for any $1\le i \ne j \le n$ and $n\ge 2$,
		$\|s_i - s_j\|\ge \Delta $, and for $s, u \in R^d$, $1 \le \ell \le p$
		and $1\le h \le k$ ($k$ is fixed),
		\begin{equation} \label{zcov}
			| \cov \{Z_\ell (s+u) , Z_\ell(s) \}| \le A/(1 + \|u\|^{d+\alpha}),
		\end{equation}
		\begin{equation} \label{fcov}
			|f_h(s)| \le A/(1 + \|s\|^{d+\alpha}).
		\end{equation}
		\item[A3.] Let $\la_1 \ge \cdots \ge \la_p \ge 0$ be the diagonal
		elements of matrix $N$ defined in (\ref{nfi}), arranged in the descending
		order. There exist integers
		$0=p_0< p_1 < \cdots <p_m=p$ for which
		\begin{align} \label{nfilamdba1a}
			& \limsup_{n \to \infty} \max_{1 \le i \le m} | \la_{p_{i-1} +1} - \la_{p_i}| =0,  \qquad
			{\rm and}\\
			\label{nfilamdba1b}
			& \liminf_{n \to \infty} \min_{1 \le i <m} |\la_{p_i} - \la_{p_i+1}| =C_1 >0,
		\end{align}
		where $m \ge 2$ is a fixed integer, and $C_1$ is a constant independent of $p$.
	\end{description}
	
	Conditions A1 and A2 are essentially the same as Assumptions 1-7 of
	\cite{bgnrv20}, though we impose only the sub-Gaussianality instead of
	requiring $Z(\cdot)$ to be normally distributed. In addition, our setting allows
	$p$ to diverge together with $n$. Condition A3 is required
	for distinguishing the columns of the mixing matrix $\Omega$ from each other.
	Those $p$ columns are completely identifiable when $p$ is fixed and $m=p$.
	Then condition (\ref{nfilamdba1a}) vanishes, and (\ref{nfilamdba1b}) ensures that
	the $p$ diagonal elements of matrix $N$ are distinct from each other (see Proposition
	\ref{prop1}). The similar conditions (i.e. with $p$ fixed) were imposed in \cite{bgnrv20}: see
	Assumptions 8 and 9 therein. Note that condition (\ref{nfilamdba1b}) cannot hold
	when $m=p\to \infty$.
	When $p\to \infty$ together with $n$, (\ref{nfilamdba1a})
	and (\ref{nfilamdba1b}) ensure that the estimated mixing matrix $\wh \Omega$ transforms
	$X(\cdot)$ into $m$ independent subvectors; see Theorem \ref{t2} below.

	Without the loss of generality, we assume that the $p$ components of $Z(\cdot)$ are
	arranged in the order such that the diagonal elements of
	matrix $N$ in (\ref{nfi}) are in the descending order. This simplifies the presentation
	of Theorem \ref{t2} substantially.
	Write $\wh W = \wh U_W \wh \Lambda_W \wh U_W^{\top}$ as its spectral decomposition,
	i.e.
	\[
	\wh \Lambda_W = \diag(\wh \la_{W,1}, \cdots, \wh \la_{W,p}),
	\]
	where $\wh \la_{W,1} \ge \cdots \ge \wh \la_{W,p}\ge 0$ are the eigenvalues of $\wh W$, and
	the columns of the orthogonal matrix $\wh U_W$ are the corresponding eigenvectors.
	Consequently,
	\begin{equation} \label{eigengammaW}
		\wh\Gamma_{\Omega}=\wh\Omega^{-1}\Omega=\wh U_W^{\top} \wh \Sigma^{-1/2} \Omega.
	\end{equation}
	Corollary \ref{corol1} below shows that $ \wh \Omega^{-1}\Omega =\wh \Gamma_{\Omega} \pcon I_p$ when $p$ is finite and $m=p$ in Condition A3. To state a more
	general result first, put $q_i = p_i-p_{i-1}$ for $i=1, \cdots, m$ (see Condition A3), and
	\begin{equation}\label{eigendecGammaw}
		\hat\Omega^{-1}\Omega=\wh\Gamma_{\Omega}=\begin{pmatrix}
			\hat\Gamma_{\Omega,11} &  \cdots &
			\hat\Gamma_{\Omega,1m}\\
			\cdots & \cdots & \cdots\\
			{\hat{\Gamma}}_{\Omega,m1} &  \cdots & {\hat{\Gamma}}_{\Omega,mm}
		\end{pmatrix},
	\end{equation}
	where submatrix $\wh \Gamma_{\Omega,ij}$ is of the size $q_i \times q_j$.
	
	\begin{theorem}\label{t2}
		Let Conditions A1-A3 hold. As $n\to \infty$ and $p=o(n)$, it holds that
		\begin{equation}\label{eigendecGamma0}
			\|{\hat{\Gamma}}_{\Omega,ii}\|=1+O_p\{n^{-1/2}p^{1/2}\}, \ \ \|{\hat{\Gamma}}_{\Omega,ii}\|_{min}=1+O_p\{n^{-1/2}p^{1/2}\} \quad 1\le i  \le m,
		\end{equation}
		
		\begin{equation}\label{eigendecGamma1w}
			\|{\hat{\Gamma}}_{\Omega,ij}\|=O_p\{n^{-1/2}p^{1/2}\}, \quad 1\le i \ne j \le m,
			\quad {\rm and}
		\end{equation}
		\begin{equation}\label{eigendecLambda1w}
			\|{\hat{\Lambda}_W}-{\Lambda}\|=O_p(n^{-1/2}p^{1/2}),
		\end{equation}
		where $\Lambda = \diag(\la_{1}, \cdots, \la_{p})$, and $\la_i$ are specified
		in Condition A3.
	\end{theorem}
	Theorem 1 implies that $ {\hat{\Gamma}}_{\Omega,ij} \pcon 0$ for
	any $i\ne j$. Hence the transformed process $ \wh \Omega^{-1} X(\cdot)=\wh\Gamma_{\Omega}Z(\cdot)$ can only be
	divided into the $m$ asymptotically independent random fields of dimensions
	$q_1, \cdots, q_m$ respectively. This is due to the lack of
	separation of the corresponding eigenvalues within each of
	those $m$ groups; see (\ref{nfilamdba1a}). On the other hand,
	Theorem 1 still holds, under some additional conditions, if the components of $Z(\cdot)$ within each
	of those $m$ groups are not independent with each other.
	Then
	this is in the spirit of the so-called
	multidimensional independent component analysis of \cite{c98}.
	In practice, one needs to identify the $m$ latent groups among the $p$ components of $\wh
	\Omega^{-1} X(\cdot)$, which can be carried out
	by adapting the procedures in Section 2.2 of \cite{cgy18}. By (\ref{eigendecLambda1w}),  $\hat{\Lambda}_W$ will indicate how
	those eigenvalues are different from each other; see Condition A3.
	
	Note that Theorem \ref{t2} holds when either $p$ is fixed and finite, or
	$p/n \to 0$ as $n\to \infty$.
	When $p$ is fixed and $m=p$ in Condition A3, all $\hat{\Gamma}_{\Omega,ij}$ reduces to a scale
	and $q_i =1$. Then Corollary \ref{corol1} below follows from Theorem 1 immediately.
	
	\begin{corollary} \label{corol1}
		Let Conditions A1-A3 hold with $m=p$, and $p$ be a fixed integer. Then as $n \to \infty$,
		$\|I_p- \wh \Omega^{-1}\Omega\|  = O_p(n^{-1/2})$.
	\end{corollary}
	
	A key condition in Corollary \ref{corol1} for identifying all the columns of
	the mixing matrix is that the eigengap defined as
	\begin{equation} \label{gapdef}
		v_{\rm gap}=
		\min_{1 \le i\ne j \le p} |\la_i - \la_j|
	\end{equation}
	remains bounded away from 0, which is implied by (\ref{nfilamdba1b})
	when $p=m$ is fixed. This condition cannot be fulfilled when $p$ diverges
	(together with $n$). To appreciate the performance of the proposed procedure
	when $p$ is large in relation to $n$, we present Theorem \ref{t3} below which
	indicates that the mixing matrix can still be estimated consistently  but
	at much slower rates when the eigengap $v_{\rm gap}$ decays to 0 provided
	$p$ diverges to $\infty$ not too fast; see Remark 1 below.
	
	\begin{description}
		\item[A4.] $\limsup_{n \rightarrow \infty}v^{-1}_{\rm gap}n^{-1/2}p^{1/2}=0$.
	\end{description}
	
	\begin{theorem} \label{t3}
		Let conditions A1, A2 and A4 hold. Denote by $\hat \gamma_{\Omega,ij}$ the $(i,j)$-th
		entry of matrix $\wh \Gamma_{\Omega} $. Then as $n, p \to
		\infty$, it holds that
		\begin{equation}\label{eigendecGamma2gapW}
			\wh \gamma_{\Omega,ij}=O_p(n^{-1/2} p^{1/2}v^{-1}_{\rm gap}|j-i|^{-1}) \quad {\rm for} \;\;
			1\le i\ne j \le p,
			\quad {\rm and}
		\end{equation}
		\begin{equation}\label{eigendecLambda2gapeW}
			\wh \gamma_{\Omega,ii}=1+O_p(n^{-1}pv^{-2}_{\rm gap} +n^{-1/2}p^{1/2}) \quad {\rm for} \;\; i=1, \cdots, p.
		\end{equation}
		Moreover, (\ref{eigendecLambda1w}) still holds.
	\end{theorem}
	
	\noindent
	{\bf Remark 1}.
	Note that $ \la_1 - \la_p \ge (p-1) v_{\rm gap}$, and, therefore,
	$v_{\rm gap} = O(p^{-1})$. Thus it follows from condition A4 that
	$p=o(n^{1/3})$, i.e. in order to fully identify the mixing matrix, $p$ cannot
	be too large in the sense that $p/n^{1/3} \to 0$.

	\section{Numerical illustration} \label{sec4}
	
	\subsection{Simulation}
	
	We illustrate the finite sample properties of the proposed method by simulation.
	We set the dimension of random fields at $p=$3 and 50, and the sample size $n$ (i.e. the number of locations) between 100 to 2000. The coordinates of those $n$ locations are drawn independently from $U(0, 50)^2$.
	Both Gaussian and  non-Gaussian random fields are used.
	Also included in the simulation is the method of \cite{bgnrv20}. For each setting, we replicate the simulation 1000 times.
	
	The $p$-variate random fields $X(\cdot)$ are generated according to  (\ref{b1})
	in which $Z_1(\cdot), \cdots, Z_p(\cdot)$ are $p$ independent random fields with either $N(0, 1)$ or $t_5$ marginal distributions, and the Matern correlation function
	\[
	\rho(s) = 2 ^{1-\kappa} \Gamma(\kappa)^{-1} (s/\phi)^\kappa B_\kappa(s/\phi),
	\]
	where $\kappa>0$ is the shape parameter, $\phi>0$ is the range parameter, $\Gamma(\cdot)$ is the Gamma function, and $B_\kappa$ is the modified Bessel function of the second kind of order $\kappa$. We set different values of $(\kappa, \phi)$ for different $Z_j$. More precisely $\kappa$'s are drawn independently from $U(0,6)$, and $\phi$'s are drawn independently from $U(0, 2)$.
	{ The mixing matrix $\Omega$ in (\ref{b1}) is set to be
		the $p\times p$ identity matrix.}
	
	To measure the accuracy of the estimation for $\Omega$, we define
	\[
	D(\Omega, \hat \Omega) = {1 \over 2 p(\sqrt{p}-1)} \sum_{j=1}^p
	\Big\{ {(\sum_{1\le i \le p} d_{ij}^2)^{1/2} \over \max_{1\le i \le p} |d_{ij}|} + { (\sum_{1\le i \le p} d_{ji}^2)^{1/2} \over \max_{1 \le i \le p} |d_{ji}|}  - 2 \Big\},
	\]
	where $d_{ij}$ is the $(i,j)$-th element of matrix $\Omega^{-1} \hat \Omega$.
	As $$p^{-1/2} \leq \max_{1 \le i \le p} |d_{ij}|\Big/\big(\sum_{1\le i \le p} d_{ij}^2
	\big)^{1/2}
	\leq 1.$$
	it holds that $D(\Omega, \hat \Omega) \in [0, 1]$, and $D(\Omega, \hat \Omega)=0$ if $\hat \Omega$ is a column permutation and/or column sign changes of $\Omega$.
	
	We set $k=10$ in (\ref{nfihatX}), and
	\begin{equation} \label{10kernels}
		f_h(s) = 1(c_{h-1} < \|s \| \le c_h), \qquad h=1, \cdots, 10,
	\end{equation}
	where $0=c_0 < c_1 < \cdots < c_{10}=\infty$ are specified such that for each $h=1, \cdots, 10$,
	$\{ (s_i, s_j): 1\le i < j\le n, \; c_{h-1} < \|s_i - s_j \| \le c_h \}$ contains
	the 10\% of the total pairs $ (s_i, s_j)$, $1\le i < j\le n$.
	
	The boxplots of $D(\Omega, \hat \Omega)$ obtained in the 1000 replications are presented in Figures \ref{dp3g}--\ref{dp50ng}.
	Estimations by the method of \cite{bgnrv20} are computed using the R-function {\tt sbss}, provided in R-package SpatialBSS.
	In addition to the multiple kernel estimation,
	we also compute the estimates with a single kernel, using each of the
	10 kernels in (\ref{10kernels}),
	For computing the multiple kernel method of \cite{bgnrv20}, we set the maximum number of iterations at 2000.
	By using a single kernel, the method of \cite{bgnrv20}  leads to almost identical estimates as those obtained by the proposed method (with the same single kernel). Therefore we omit the detailed results.
	
	Figures \ref{dp3g} -- \ref{dp50ng} indicate clearly that
	both the methods with multiple kernels outperform most of those
	with a single kernel, and the proposed method
	outperforms the multiple kernel method of \cite{bgnrv20} especially when $p$ is large (i.e. $p=50$).
	The proposed method with multiple kernels performs about the same as that with the best single kernel (i.e. Kernel 1 $f_1(\cdot)$). The accuracy of estimation improves with the increase in the number of observations $n$, which can be  seen as a decrease in $D(\Omega, \hat \Omega)$ in Figures \ref{dp3g}--\ref{dp50ng}. Among all single kernel methods, those using kernel $f_1$ perform  the best, as those estimations include the 10\% nearest locations. Indeed the Matern correlation is the strongest at the smallest distance. On the other hand, the performances for the
	Gaussian and the non-Gaussian random fields are
	about the same. See  Figures \ref{dp3g} \& \ref{dp3ng}, and Figures \ref{dp50g} \& \ref{dp50ng}.
	
	The iterative algorithm for implementing the multiple kernel method of \cite{bgnrv20} is to solve a nonlinear optimization problem with $p^2$ parameters. When $p=50$, it failed to converge within the 2000 iterations in some of the 1000 simulation replications. The numbers of failures with $n=$100, 500, 1000 and 2000 are, respectively, $3,1,2$ and $1$ for the Gaussian random fields,  and $6,3,3$ and $1$ for the non-Gaussian random fields. We only include the results from the converged replications in the figures.
	
	\begin{figure}[H]
		\vskip 0.2in
		\begin{center}
			\includegraphics[height=2.4in, width=5in]{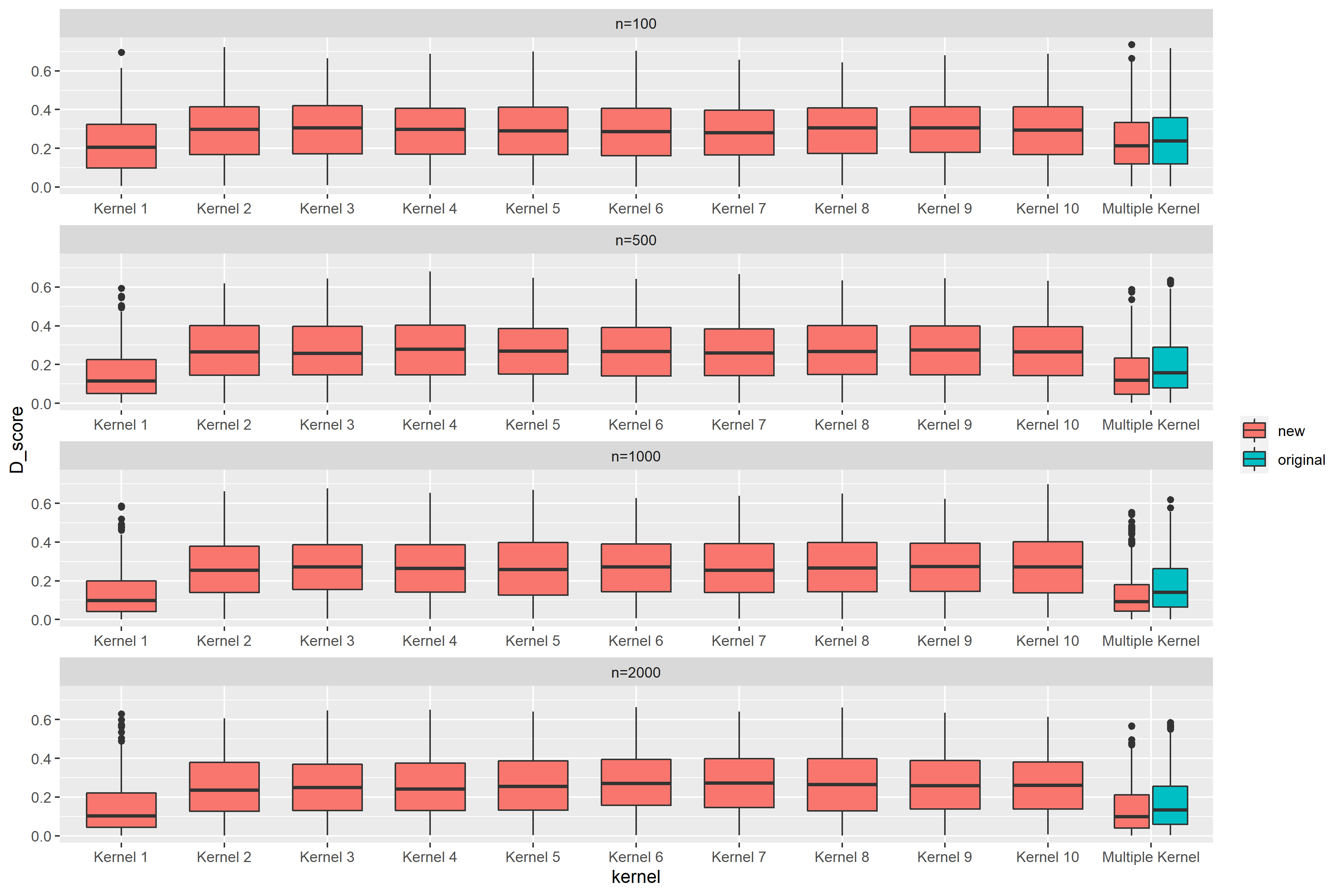}
			\caption{Boxplots of $D(\Omega, \hat \Omega)$ for the proposed method using the 10 kernels (new) in (\ref{10kernels}), or each of those 10 kernels (Kernel 1, $\cdots$, Kernel 10), and the method of \cite{bgnrv20} using the 10 kernels (original) in a simulation with 1000 replications for the Gaussian random fields. The number of observations $n$ is  100, 500, 1000 or 2000 (from top to bottom), and the dimension of random fields is $p=3$. }
			\label{dp3g}
		\end{center}
		\vskip -0.2in
	\end{figure}
	
	\begin{figure}[H]
		\vskip 0.2in
		\begin{center}
			\includegraphics[height=2.4in, width=5in]{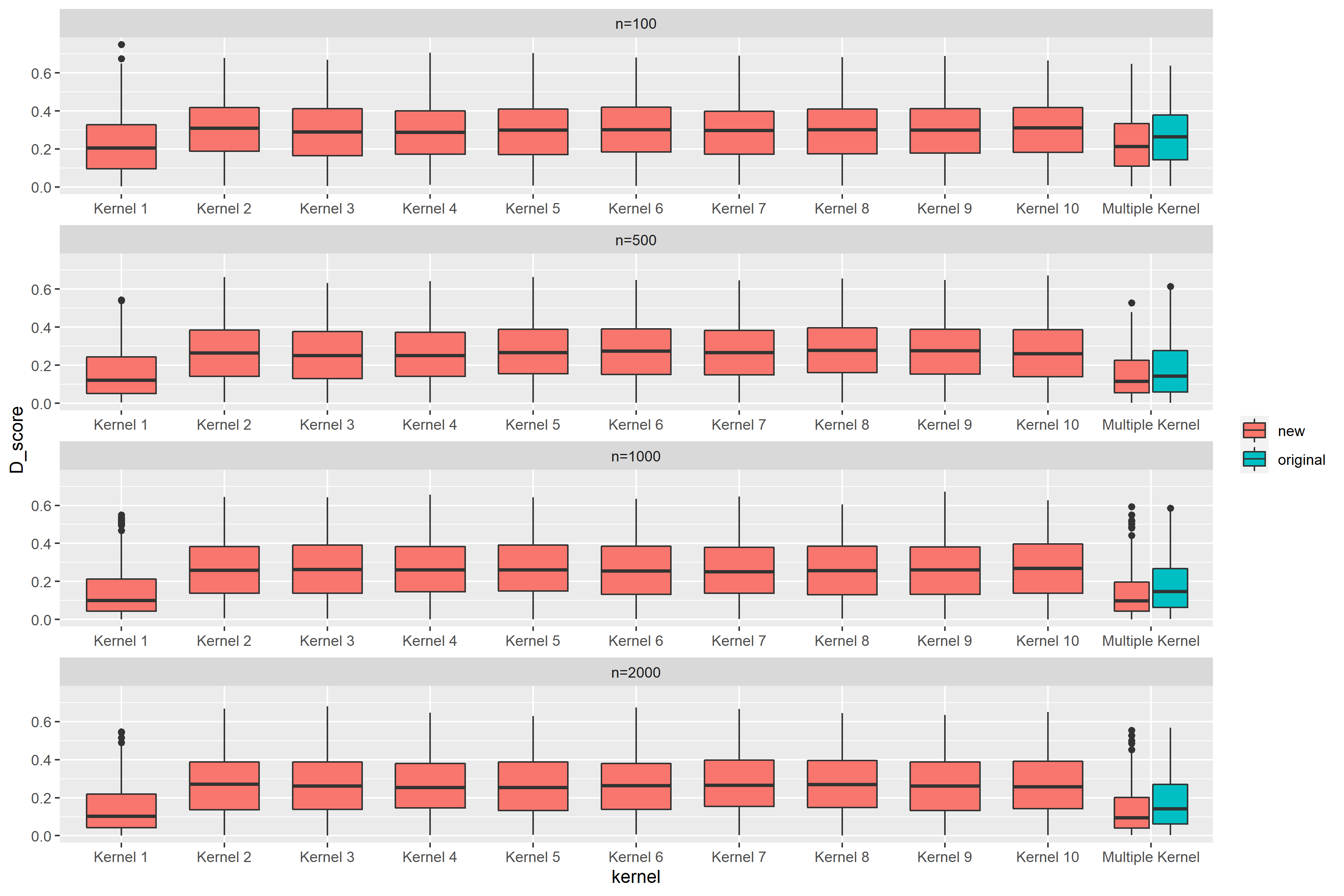}
			\caption{Boxplots of $D(\Omega, \hat \Omega)$ for the proposed method using the 10 kernels (new) in (\ref{10kernels}), or each of those 10 kernels (Kernel 1, $\cdots$, Kernel 10), and the method of \cite{bgnrv20} using the 10 kernels (original) in a simulation with 1000 replications for the non-Gaussian random fields. The number of observations $n$ is  100, 500, 1000 or 2000 (from top to bottom), and the dimension of random fields is
				$p=~3$.}
			\label{dp3ng}
		\end{center}
		\vskip -0.2in
	\end{figure}
	
	\begin{figure}[H]
		\vskip 0.2in
		\begin{center}
			\includegraphics[height=2.4in, width=5in]{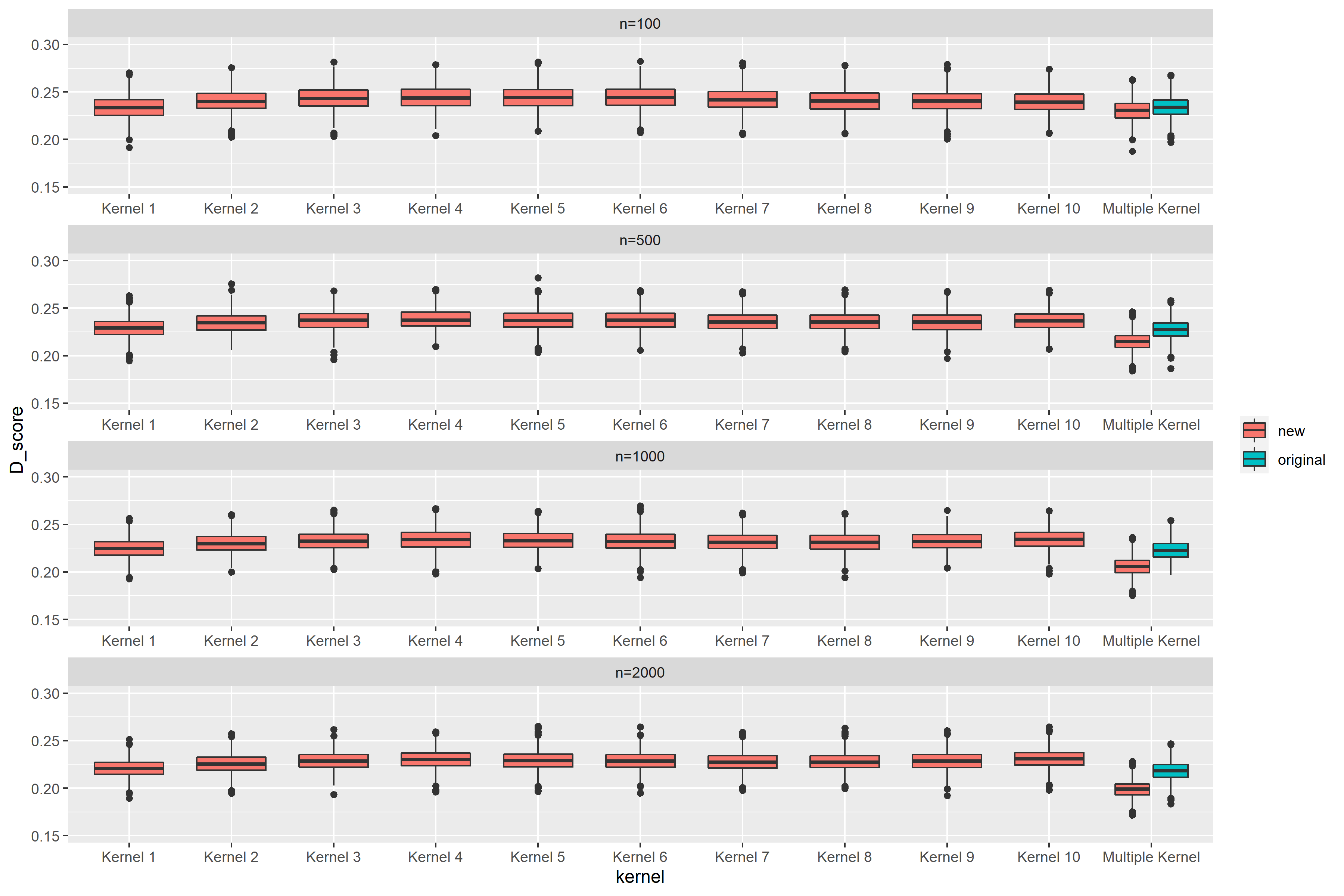}
			\caption{Boxplots of $D(\Omega, \hat \Omega)$ for the proposed method using the 10 kernels (new) in (\ref{10kernels}), or each of those 10 kernels (Kernel 1, $\cdots$, Kernel 10), and the method of \cite{bgnrv20} using the 10 kernels (original) in a simulation with 1000 replications for the Gaussian random fields. The number of observations $n$ is  100, 500, 1000 or 2000 (from top to bottom), and the dimension of random fields is $p=50$.}
			\label{dp50g}
		\end{center}
		\vskip -0.2in
	\end{figure}
	
	\begin{figure}[H]
		\vskip 0.2in
		\begin{center}
			\includegraphics[height=2.4in, width=5in]{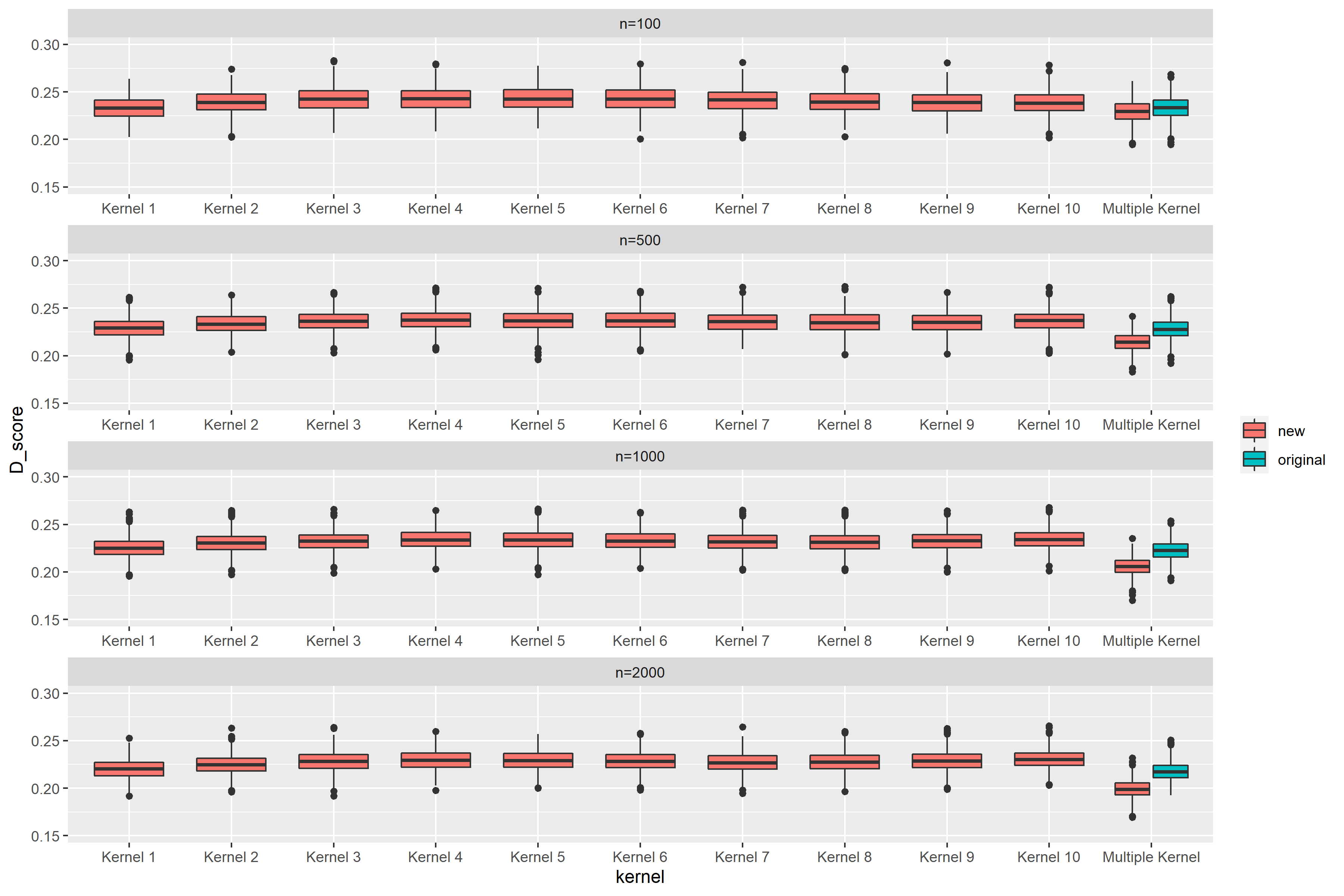}
			\caption{Boxplots of $D(\Omega, \hat \Omega)$ for the proposed method using the 10 kernels (new) in (\ref{10kernels}), or each of those 10 kernels (Kernel 1, $\cdots$, Kernel 10), and the method of \cite{bgnrv20} using the 10 kernels (original) in a simulation with 1000 replications for the non-Gaussian random fields. The number of observations $n$ is  100, 500, 1000 or 2000 (from top to bottom), and the dimension of random fields is $p=50$.}
			\label{dp50ng}
		\end{center}
		\vskip -0.2in
	\end{figure}
	
	The estimated eigengaps for the proposed method for the Gaussian
	random fields are presented in Figures \ref{egp3g} and \ref{egp50g}. As $n$ increases, the eigengap also increases. Under low-dimensional setting $p=3$, the estimates based on single kernel $f_1$ entail the largest eigengaps and the smallest estimation errors $D(\Omega, \hat \Omega)$ (see also Theorem \ref{t3}).
	However when $p=50$, using the multiple kernels leads to
	the largest eigengaps and the smallest estimation errors.
	The patterns with the non-Gaussian random fields are similar  and  not reported here to save  space.

	\begin{figure}[H]
		\vskip 0.2in
		\begin{center}
			\includegraphics[height=2.5in, width=5in]{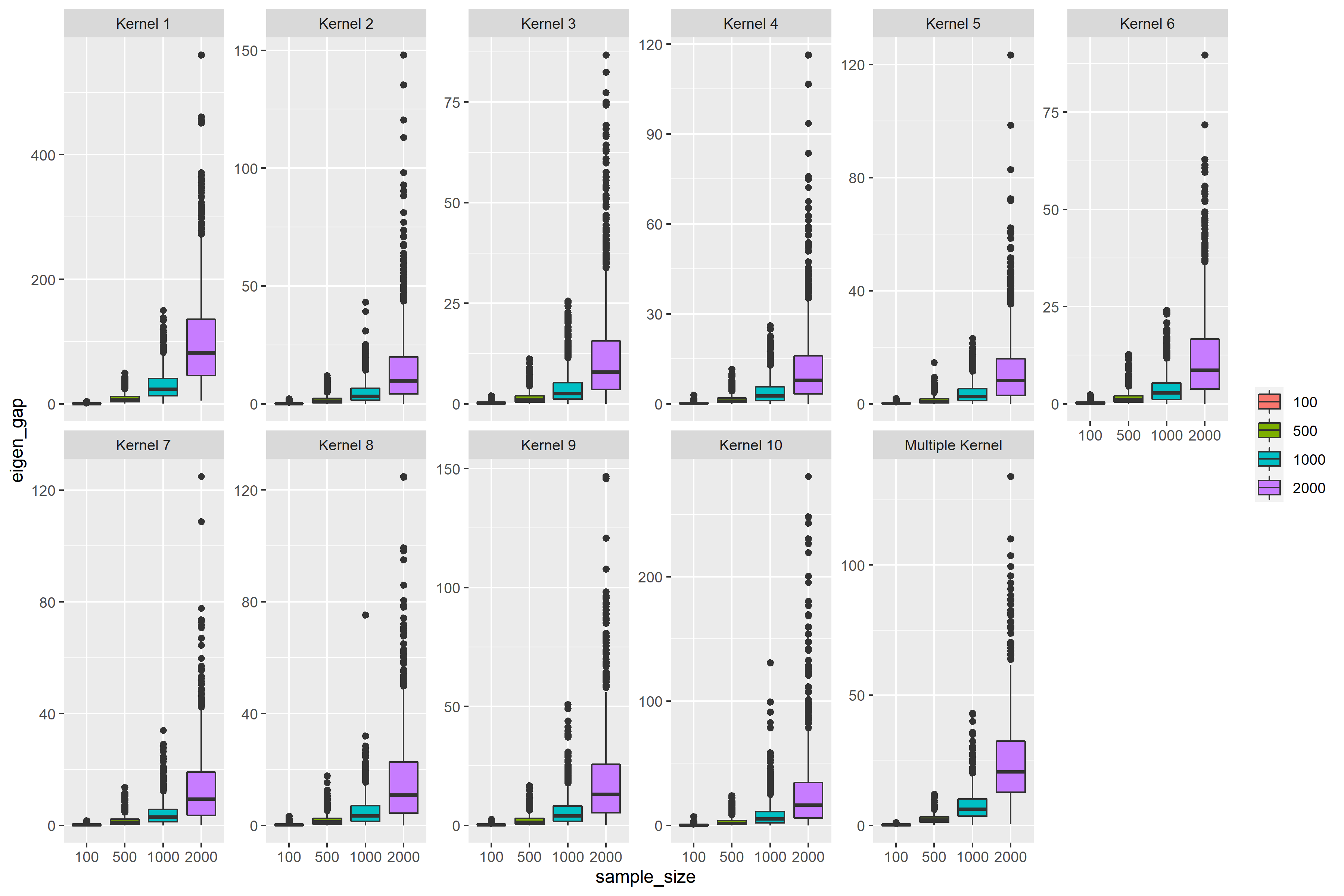}
			\caption{Boxplots of the estimated eigengaps of the proposed method using the 10 kernels (Multiple kernels) in (\ref{10kernels}), or each of those 10 kernels (Kernel 1, $\cdots$, Kernel 10) for the Gaussian random fields. Number of observations $n$ is set at 100, 500, 1000 and 2000, the dimension of random fields is $p=3$.}
			\label{egp3g}
		\end{center}
		\vskip -0.2in
	\end{figure}

	\begin{figure}[H]
		\vskip 0.2in
		\begin{center}
			\includegraphics[height=2.5in, width=5in]{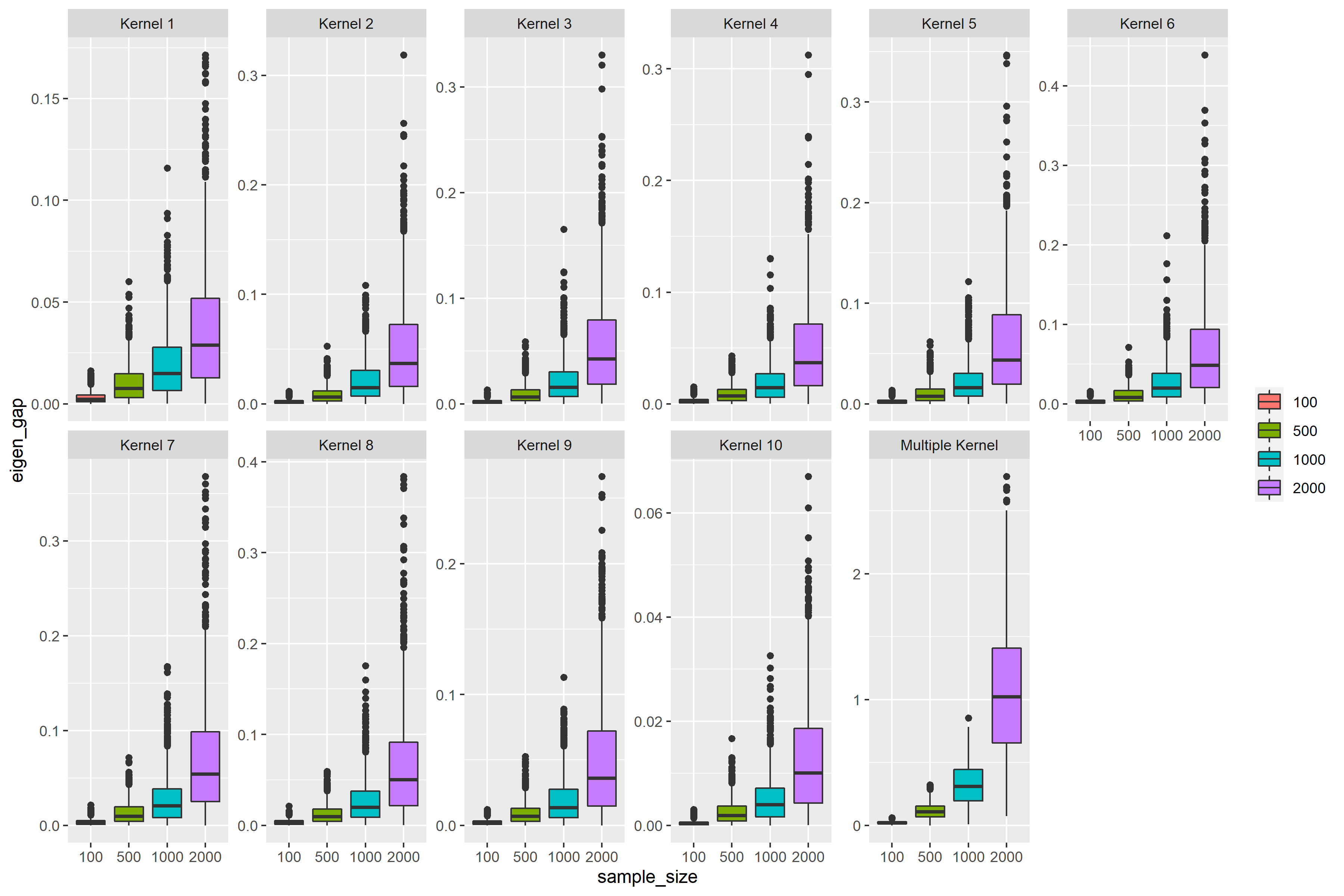}
			\caption{Boxplots of the estimated eigengaps of the proposed method using the 10 kernels (Multiple kernels) in (\ref{10kernels}), or each of those 10 kernels (Kernel 1, $\cdots$, Kernel 10) for the Gaussian random fields. Number of observations $n$ is set at 100, 500, 1000 and 2000, the dimension of random fields is $p=50$.}
			\label{egp50g}
		\end{center}
		\vskip -0.2in
	\end{figure}

	\subsection{A real data example}
	
	We apply the proposed method to the moss data from the Kola project in the R package {\tt StatDa} (See \cite{f15}). The data consists of chemical elements discovered in terrestrial moss at the 594 locations in northern Europe; see the map in Fig.D.1 of \cite{bgnrv20}.
	% see the map in figure \ref{map}.
	More information on the data is presented in \cite{rfgd08}.
	Following the lead of \cite{nofr15} and \cite{bgnrv20}, we apply the so-called isometric-log-ratio transformation to the 31 compositional chemical elements in the data. The transformed data are  used in our analysis with $n=594$ and $p=30$. We standardize the data first such that the sample mean is 0 and the sample variance is $I_{30}$.
	
	We apply the proposed estimation method with 10 kernels specified as in (\ref{10kernels}). The scores of the first six independent components (IC), corresponding to the six largest eigenvalues of $\hat W$ (see table \ref{egreal1}), are plotted in Figure \ref{NIC}; showing some interesting spatial patterns. For example, the 1st IC can be viewed as a contrast between the locations in the west and those in the east, and the 2nd IC is that between the north and the south. Figure \ref{corr} displays the absolute correlation coefficients between the first twelve ICs and those obtained in \cite{nofr15} which was referred as `gold standard' by \cite{bgnrv20}. While the ICs derived from the two methods differ from each other, the two sets of ICs correlate with each other significantly. For example the correlation between the 1st IC derived from our new method and the 2nd IC obtained in \cite{nofr15} is 0.92. Note that the `gold standard'
	estimation was obtained using the kernel specified with the relevant subject knowledge. In contrast our estimation is based on
	the multiple kernels defined generically in (\ref{10kernels}).
	
	The six largest eigenvalues of $\hat W$ are listed in Table \ref{egreal1}. The eigengaps $\Delta_i= \hat \la_{i-1} - \hat \la_i$ for $i=7, \cdots, 30$ are plotted in Figure \ref{egreal2}. It is clear that the eigengaps among the 13 largest eigenvalues are large. Based on Theorem \ref{t2}, we have
	\begin{equation}\label{eigendecGammawrealdata}
		\hat\Omega^{-1}\Omega=\wh\Gamma_{\Omega}=\begin{pmatrix}
			\hat\Gamma_{\Omega,aa} &
			\hat\Gamma_{\Omega,ab}\\
			{\hat{\Gamma}}_{\Omega,ba} &   {\hat{\Gamma}}_{\Omega,bb}
		\end{pmatrix},
	\end{equation}
	where $\hat\Gamma_{\Omega,aa}$ is a $12 \times 12$ matrix satisfying $\|\hat\Gamma_{\Omega,aa}-I_{12}\|=O_p(n^{-1/2}p^{1/2})$.  Theorem \ref{t2} also shows that $\|\hat\Gamma_{\Omega,ab}\|=O_p(n^{-1/2}p^{1/2})$, $\|\hat\Gamma_{\Omega,ba}\|=O_p(n^{-1/2}p^{1/2})$ and $\|\hat\Gamma_{\Omega,bb}\|=1+O_p(n^{-1/2}p^{1/2})$. Thus, we
	are reasonably confident that the estimated first 12 ICs are reliable. Moreover, we rewrite $\hat\Omega^\top\hat\Omega$ as
	\begin{equation}\label{eigendecGammawrealdata22}
		\hat U_{W}^\top\hat \Sigma\hat U_{W}=\hat\Omega^\top\hat\Omega=\begin{pmatrix}
			\hat\Omega_{aa} &
			\hat\Omega_{ab}\\
			{\hat{\Omega}}_{ba} &   {\hat{\Omega}}_{bb}
		\end{pmatrix},
	\end{equation}
	where $\hat\Omega_{aa}$ is a $12 \times 12$ matrix.
	We gain $tr(\hat\Omega_{aa})=6.62$ and $tr(\hat\Omega^\top\hat\Omega)=8.89$ by calculating. Thus,  the major variation of the 30 variables are largely reflected by the 12 largest ICs.

	\begin{figure}
		\centering\centering
		\begin{tabular}{@{}cc@{}}
			\includegraphics[scale=0.4]{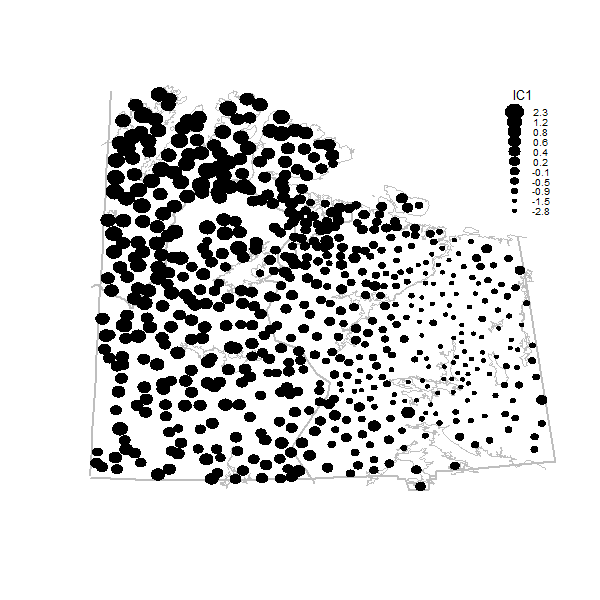} & \includegraphics[scale=0.4]{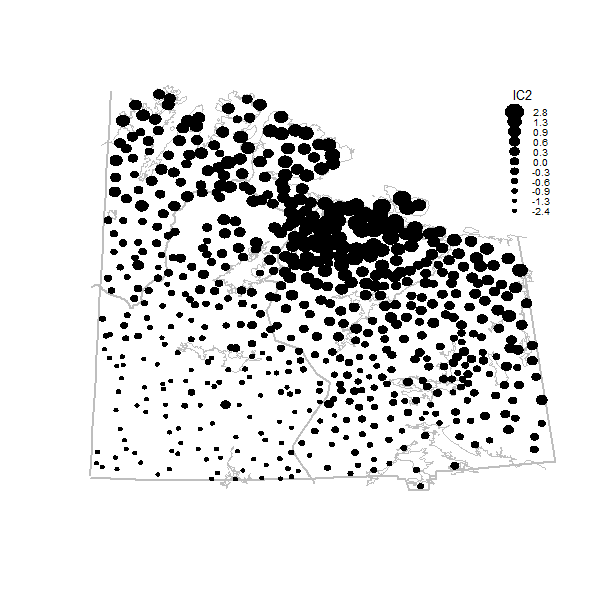}\\[-10ex]
			\includegraphics[scale=0.4]{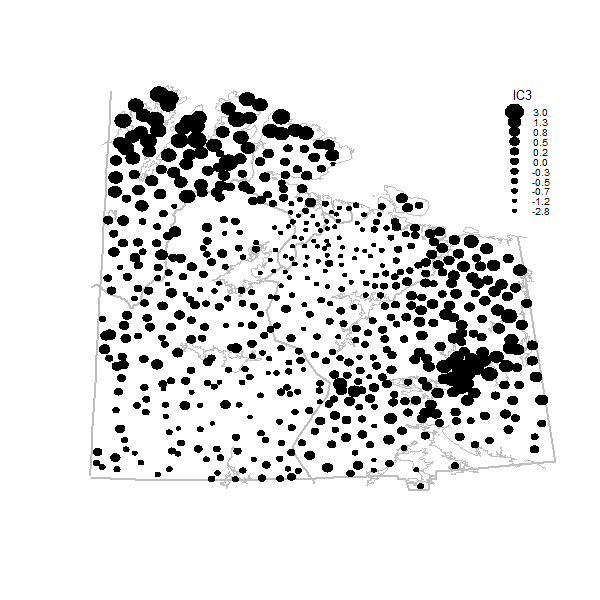} & \includegraphics[scale=0.4]{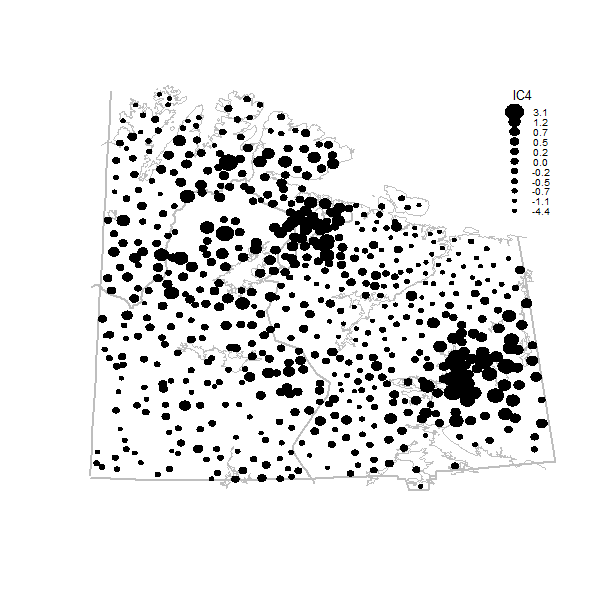}\\[-10ex]
			\includegraphics[scale=0.4]{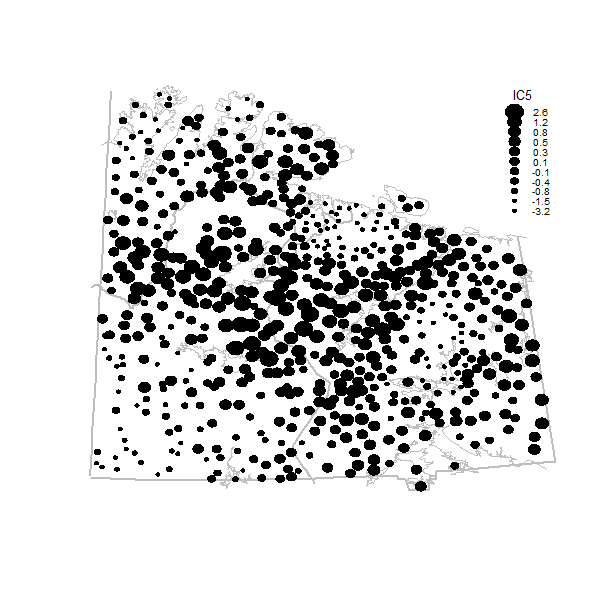} & \includegraphics[scale=0.4]{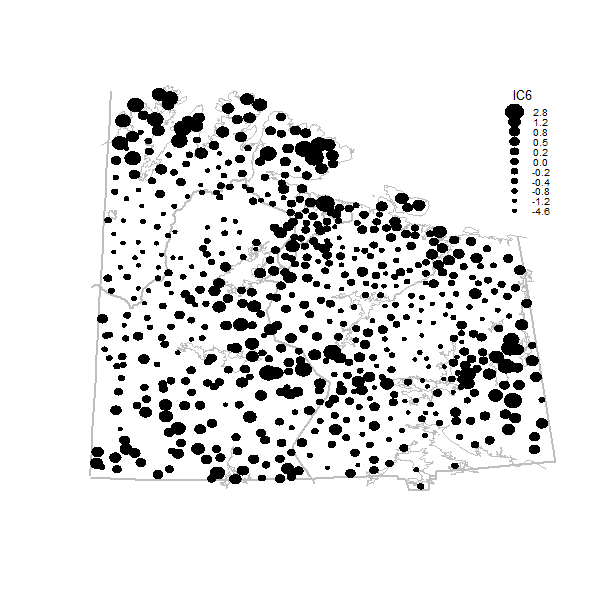}\\[-5ex]
		\end{tabular}
		\caption{The scores of the first six independent components over the 594 observation
			locations.}
		\label{NIC}
	\end{figure}
	
	\begin{figure}
		\centering
		\centerline{\includegraphics[scale=0.5]{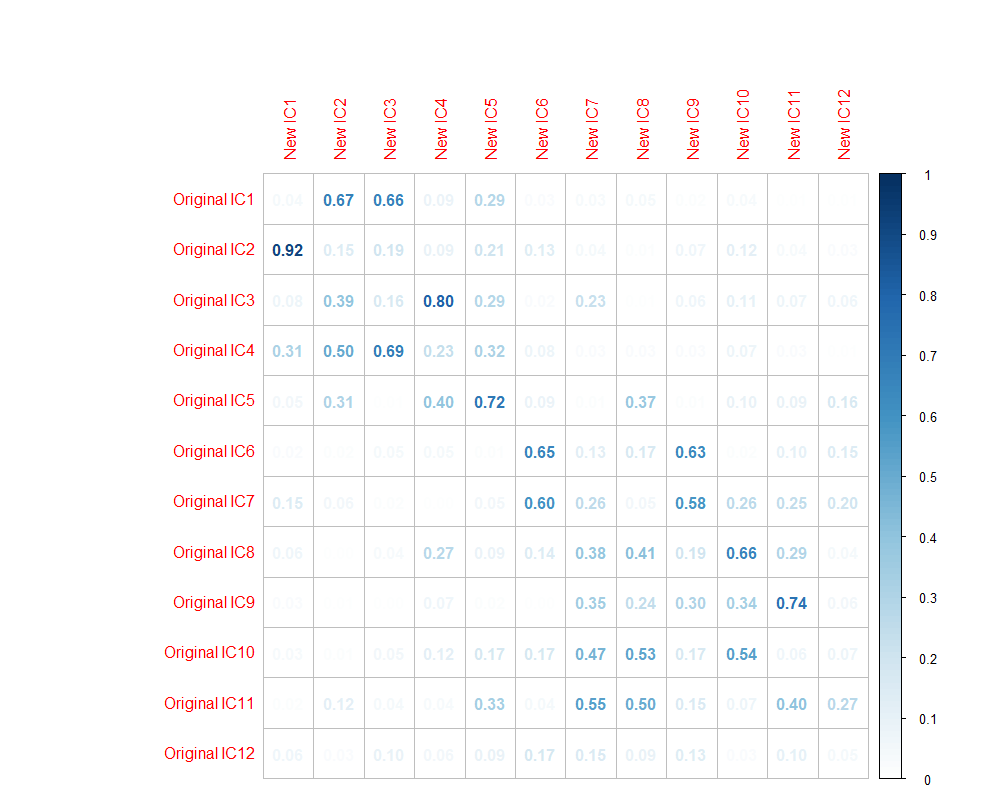}}
		\caption{The absolute correlation coefficients between the first 12 independent
			components derived from the proposed method (New) and those obtained in
			\cite{nofr15} (Original).}
		\label{corr}
	\end{figure}
	
	\begin{table*}[]
		\centering
		\caption{The six largest eigenvalues of $\hat W$ (with $k=10$) for the real data example.} \label{egreal1}
		\begin{tabular}{|c|cccccc|}
			\hline
			$i$ & 1 & 2 & 3 & 4 & 5 & 6 \\
			\hline
			$\hat\lambda_i$ & 1136.50 & 877.59 & 444.21 & 161.34 & 126.16 & 81.13 \\
			\hline
		\end{tabular}
	\end{table*}
	\begin{figure}
		\centering\centering
		\centerline{\includegraphics[scale=0.65]{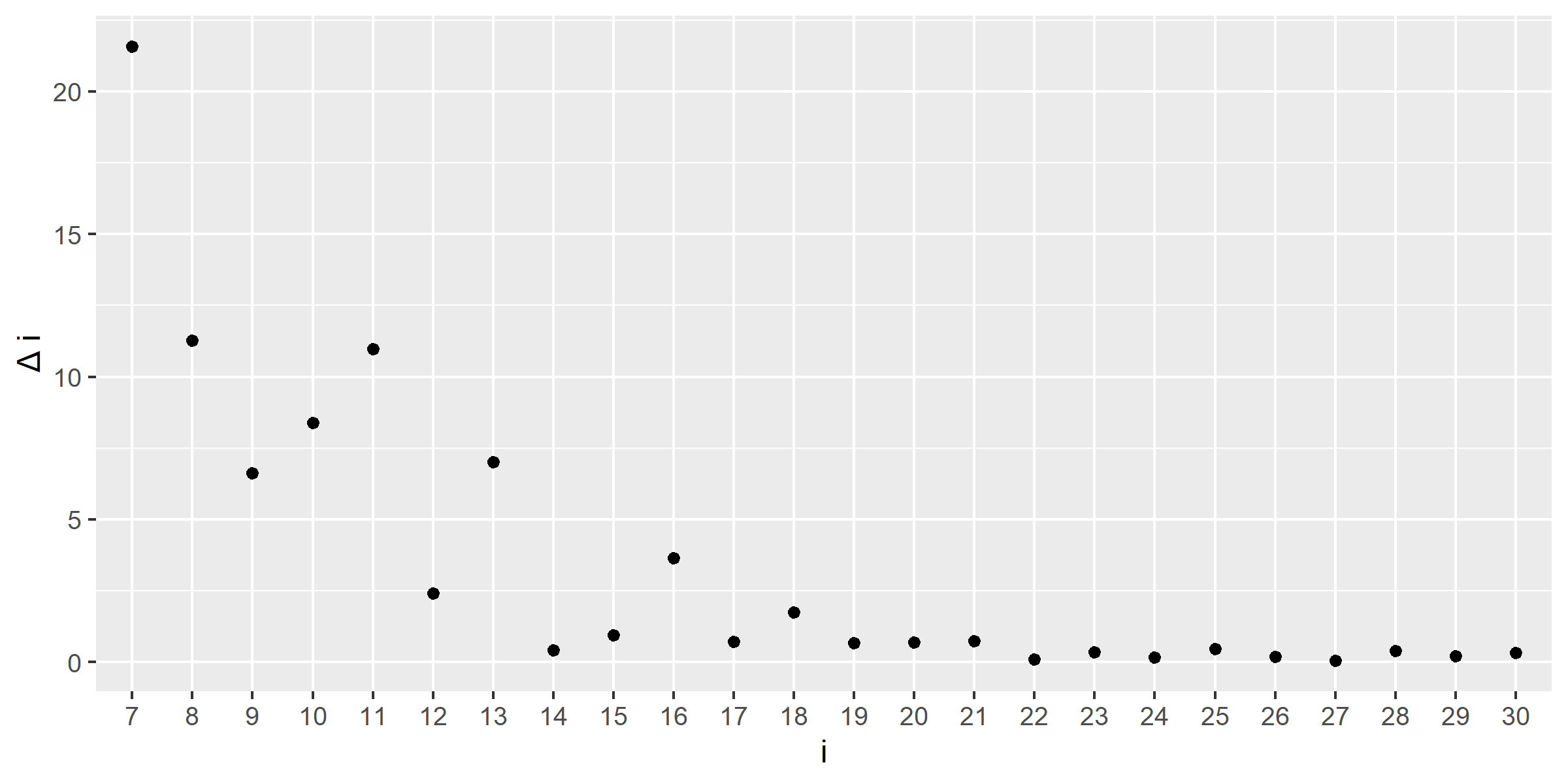}}
		\caption{The estimated eigengaps $\Delta_i=\hat \la_{i-1} - \hat \la_{i}$ for $i=7, \cdots, 30$ on real data example from proposed method with multiple kernel.}
		\label{egreal2}
	\end{figure}

	\section{Proofs} \label{sec5}
\subsection{Some useful lemmas}
$C_0$ which is defined in Condition A1  and $A$ which is defined in Condition A2 are two important notations in our proofs. Without loss of generality, we assume that $C_0 \leq A$. It means that
\begin{equation} \label{subgaussianassuA}
	\sup_{\beta \geq 1, 1\leq i \leq p}\beta^{-1/2}\{E|Z_i(s)|^\beta\}^{1/\beta} \leq A.
\end{equation}
Thus, any fixed moment  of $Z_g(s)$ can be bounded by a constant only depending on $A$.

Let $Z$ be the $p \times n$ matrix with $(Z_i(s_1),\cdots,Z_i(s_n))=Z^i$ as
its $i$-th row.

\begin{lemma} \label{l0z}
	Let conditions A1 and A2 hold, and $p=o(n)$. Then there exists
	$\lambda_{max}$ depending only on $A$ such that
	\begin{equation}\label{nfilamdba1max}
		\max_{1 \leq g \leq p}\lambda_g \leq \lambda_{max} < \infty.
	\end{equation}
\end{lemma}
\begin{proof} For any $g=1, \cdots, p$,  (\ref{nfi}) implies that
	\begin{eqnarray}\label{nfi222}
		&&\lambda_g=\frac{1}{k}\sum_{h=1}^k\sum_{u=1}^pE[\frac{1}{n}\sum_{i,j=1}^nf_h(s_i-s_j) \tilde Z_g(s_i) \tilde Z_u(s_j)]^2\\
		\nonumber=&&\frac{1}{k}\sum_{h=1}^k\sum_{u \neq
			g}E[\frac{1}{n}\sum_{i,j=1}^nf_h(s_i-s_j) \tilde Z_g(s_i) \tilde Z_u(s_j)]^2\;
		+\; \frac{1}{k}\sum_{h=1}^kE[\frac{1}{n}\sum_{i,j=1}^nf_h(s_i-s_j) \tilde Z_g(s_i) \tilde Z_g(s_j)]^2.
	\end{eqnarray}
	We consider the first part $u \neq g$ for each $h$,
	\begin{eqnarray*}
		&&\sum_{u \neq g}E[\frac{1}{n}\sum_{i,j=1}^nf_h(s_i-s_j) \tilde Z_g(s_i) \tilde Z_u(s_j)]^2\\
		=&&\sum_{u \neq g}\frac{1}{n^2}\sum_{i,j,\tilde{i},\tilde{j}=1}^nf_h(s_i-s_j)f_h(s_{\tilde{i}}-s_{\tilde{j}})E [ \tilde Z_g(s_i) \tilde Z_u(s_j) \tilde Z_g(s_{\tilde{i}}) \tilde Z_u(s_{\tilde{j}})]\\
		=&&\sum_{u \neq g}\frac{1}{n^2}\sum_{i,j,\tilde{i},\tilde{j}=1}^nf_h(s_i-s_j)f_h(s_{\tilde{i}}-s_{\tilde{j}})E[ \tilde Z_g(s_i) \tilde Z_g(s_{\tilde{i}})]E[ \tilde Z_u(s_j) \tilde Z_u(s_{\tilde{j}})]\\
		\leq &&\sum_{u \neq g}\frac{1}{n^2}\sum_{i,j,\tilde{i},\tilde{j}=1}^n\frac{A}{1+\|s_i-s_j\|^{d+\alpha}}\frac{A}{1+\|s_{\tilde{i}}-s_{\tilde{j}}\|^{d+\alpha}}\frac{A}{1+\|s_{i}-s_{\tilde{i}}\|^{d+\alpha}}\frac{A}{1+\|s_{j}-s_{\tilde{j}}\|^{d+\alpha}}.
	\end{eqnarray*}
	The last inequality is from (\ref{zcov}) and (\ref{fcov}).
	This, together with $p=o(n)$ and $\|s_i-s_j\|\geq \triangle$ for all $n \geq 2$ and $1 \leq i \neq j \leq n$, implies that
	\begin{eqnarray}\label{nfi223}
		\frac{1}{k}\sum_{h=1}^k\sum_{u \neq g}E[\frac{1}{n}\sum_{i,j=1}^nf_h(s_i-s_j) \tilde Z_g(s_i) \tilde Z_u(s_j)]^2=O(A^4n^{-1}p)=o(1).
	\end{eqnarray}
	Thus we only need to consider $E[\frac{1}{n}\sum_{i,j=1}^nf_h(s_i-s_j) \tilde Z_g(s_i) \tilde Z_g(s_j)]^2$. Since $Z^g=(Z_g(s_1),\cdots,Z_g(s_n))$ and
	\begin{eqnarray}\label{zgtiledz}
		(\tilde Z_g(s_1),\cdots,\tilde Z_g(s_n))=Z^g[I_n-n^{-1}1_{n \times n}].
	\end{eqnarray}
	We can rewrite it as
	$E(\frac{1}{n}Z^g[I_n-n^{-1}1_{n \times n}]T_h[I_n-n^{-1}1_{n \times n}](Z^g)^\top)^2$,
	where $T_h$ is a $n \times n$ matrix with the $(i,j)$th entry
	$f_h(s_i-s_j)/2+f_h(s_j-s_i)/2$.
	Note that $\frac{1}{n}Z^g[I_n-n^{-1}1_{n \times n}]T_h[I_n-n^{-1}1_{n \times n}](Z^g)^\top$ is a quadratic form  and $Z_g(s)$ is
	a sub-Gaussian process. (\ref{fcov}) implies that $\|T_h\| \leq \tilde{C}$,
	where $\tilde{C}$ only depends on $A$. These, together with (\ref{zcov}), imply that there exists a positive constant $\tilde{C}_1$ depending only on $A$ such that
	\begin{eqnarray*}
		\frac{1}{k}\sum_{h=1}^kE[\frac{1}{n}\sum_{i,j=1}^nf_h(s_i-s_j) \tilde Z_g(s_i) \tilde Z_g(s_j)]^2 \leq \tilde{C}_1.
	\end{eqnarray*}
	This, together with (\ref{nfi222})- (\ref{nfi223}), implies that
	$
	\lambda_g \leq 2\tilde{C}_1,
	$
	for any $1 \leq g \leq p$.
	We complete the proof.
\end{proof}

\begin{lemma}\label{l1z}
	Let conditions A1 and A2 hold. For any $n \times n$ non-random symmetric matrix $Q$
	with bounded $\|Q\|$,   there exists a constant $C>0$  depending only on
	$A$ and $\lambda_{max}$ for which
	\begin{equation}\label{Zasy0}
		\max_{1 \leq g,u \leq p} var[\frac{1}{n}\sum_{i,j=1}^nQ_{ij}Z_g(s_i)Z_{u}(s_j)] \leq C\|Q\|^2n^{-1}.
	\end{equation}
	Here $Q_{ij}$ is the $(i,j)-$th entry of $Q$.
\end{lemma}

\begin{proof}
	When $g \neq u$, from the independence between $Z_g(s_i)$ and $Z_{u}(s_j)$  we have
	\begin{eqnarray*}
		&&var[\frac{1}{n}\sum_{i,j=1}^nQ_{ij}Z_g(s_i)Z_{u}(s_j)]\\
		=&&n^{-2}\sum_{i_1,j_1,i_2,j_2=1}^n Q_{i_1j_1}Q_{i_2j_2}E[Z_g(s_{i_1})Z_{u}(s_{j_1})Z_g(s_{i_2})Z_{u}(s_{j_2})]\\
		=&&n^{-2}\sum_{i_1,j_1,i_2,j_2=1}^n Q_{i_1j_1}Q_{i_2j_2}E[Z_g(s_{i_1})Z_g(s_{i_2})]E[Z_{u}(s_{j_1})Z_{u}(s_{j_2})]\\
		\leq&&n^{-2}\sum_{i_1,j_1,i_2,j_2=1}^n Q_{i_1j_1}Q_{i_2j_2}\frac{A}{1+\|s_{i_1}-s_{i_2}\|^{d+\alpha}}\frac{A}{1+\|s_{j_1}-s_{j_2}\|^{d+\alpha}}\\
		\leq&& C\|Q\|^2n^{-1}.
	\end{eqnarray*}
	The first inequality is from (\ref{zcov}) and (\ref{fcov}). The second inequality is from $\|s_i-s_j\|\geq \triangle$ for all $n \geq 2$ and $1 \leq i \neq j \leq n$.
	When $g=u$, we note that $\frac{1}{n}\sum_{i,j=1}^nQ_{ij}Z_g(s_i)Z_{g}(s_j)$ is a
	quadratic form and $Z_g(s)$ is a sub-Gaussian process.  This completes the proof.
\end{proof}
\begin{lemma}\label{lemCZ}
	Let conditions A1 and A2 hold, and $p=o(n)$. Then there exists a positive constant $C_A$ depending only on $A$ such that
	\begin{equation}\label{CZasy}
		\lim_{n \rightarrow \infty}P(n^{-1}\|Z\|^2 \leq C_A) =1.
	\end{equation}

\end{lemma}
\begin{proof}
	
	For any fixed $1 \times n$ unit vector $x=(x_1,\cdots,x_{n})$, we denote $xZ^\top$
	by $ z(x)=\Big(z_{1}(x),\cdots,z_{p}(x)\Big)$. Since $Z_1(\cdot), \cdots, Z_p(\cdot)$ are
	independent, the elements of $z(x)$ are independent. (\ref{zcov}) implies that $\max_{1 \leq j \leq p}Ez_{j}^2(x) \leq \tilde{C}_A$ where $\tilde{C}_A$ only depends on $A$.
	\begin{eqnarray*}
		xZ^\top Zx^\top=\sum_{j=1}^{p} [z_{j}^2(x)-Ez_{j}^2(x)]+\sum_{j=1}^{p}Ez_{j}^2(x) \leq \sum_{j=1}^{p} [z_{j}^2(x)-Ez_{j}^2(x)]+p\tilde{C}_A.
	\end{eqnarray*}
	By the sub-Gaussian property of $Z(s)$, we can conclude that for any fixed $1 \times p$ unit vector $x$ and any $c>0$ there exists $\tilde{C}_{A,1}$ depending only on $A$ and $c$ such that
	\begin{eqnarray}\label{HhJ1bound2zz}
		P\Big(\|xZ^\top\|^2 > \tilde{C}_{A,1}(n+p)\Big) \leq c\exp(-5(n+p)).
	\end{eqnarray}
	As we know, the unit Euclidean sphere $S^{n-1}$ consists of all
	$n$-dimensional unit vectors $x$.    Unfortunately   the cardinality of
	$S^{n-1}$ is uncountable cardinal number.
	We can't use
	(\ref{HhJ1bound2zz}) to derive an upper bound of $\|Z\|^2$ directly.
	Thus we introduce a method based on nets to control $\|Z\|^2$. The basic
	idea is as follows. We define a subset of $S^{n-1}$ as $S_{\varepsilon}$
	satisfying $\max_{x \in S^{n-1}}\min_{y \in S_{\varepsilon}}\|x-y\| \leq
	\varepsilon$. $S_{\varepsilon}$ is a so-called net of  $S^{n-1}$ and the
	cardinality of $S_{\varepsilon}$ is bounded by
	$(1+2\varepsilon^{-1})^{n}$. Thus we can control $\max_{y \in
		S_{\varepsilon}}\|Zy^\top\|$ in probability by (\ref{HhJ1bound2zz}). Finally,
	we can control the difference between $\max_{y \in S_{\varepsilon}}\|Zy^\top\|$ and $\max_{x \in
		S^{n-1}}\|Zx^\top\|$.

	Let $S_{\varepsilon}$ be a subset of $S^{n-1}$. For any $x \in S^{n-1}$, there exists $\tilde{x} \in S_{\varepsilon}$ such that $\|\tilde{x}-x\| \leq \varepsilon$.
	This, together with (\ref{HhJ1bound2zz}) and $|S_{\varepsilon}| \leq (1+2\varepsilon^{-1})^{n}$, implies that
	\begin{eqnarray}\label{HhJ1bound3zz00}
		P\Big(\max_{\tilde{x} \in
			S_{1/2}}\|Z\tilde{x}^\top\|^2 > \tilde{C}_{A,1}(n+p)\Big) \leq c|S_{1/2}|\exp(-5n-5p) \leq c5^{n} \exp(-5n-5p).
	\end{eqnarray}

	Then if $\|Zx^\top\|=\|Z\|$, there exists $\tilde{x} \in S_{\varepsilon}$ such that
	\begin{eqnarray*}
		\|Z\tilde{x}^\top\| \geq \|Zx^\top\|- \|Z(\tilde{x}-x)^\top\|  \geq \|Z\|-\varepsilon\|Z\|=(1-\varepsilon)\|Z\|.
	\end{eqnarray*}
	Let $\varepsilon=1/2$,
	\begin{eqnarray*}
		\|Z\|^2 \leq 4 \max_{\tilde{x} \in S_{1/2}} \|Z\tilde{x}^\top\|^2.
	\end{eqnarray*}
	This, together with (\ref{HhJ1bound3zz00}), implies that
	\begin{eqnarray}\label{HhJ1bound3zz}
		P\Big(\|Z\|^2 > 4\tilde{C}_{A,1}(n+p)\Big) \leq c|S_{1/2}|\exp(-5n-5p) \leq c5^{n} \exp(-5n-5p).
	\end{eqnarray}
	Then (\ref{CZasy}) is implied by (\ref{HhJ1bound3zz}) and $p=o(n)$.
\end{proof}
\begin{definition}
	\begin{equation} \label{nfihat}
		\wh N =  {1 \over k}\sum_{h=1}^k \big\{ {1 \over n} \sum_{i,j=1}^n f_h(s_i - s_j)
		\tilde Z(s_i) \tilde Z(s_j)^\top \big\}\big\{ {1 \over n} \sum_{i,j=1}^n f_h(s_i - s_j)
		\tilde Z(s_i) \tilde Z(s_j)^\top \big\}^\top.
	\end{equation}
\end{definition}
\begin{lemma}\label{t1}
	Let conditions A1 and A2 hold, and $p=o(n)$. Let $M_{gu}$ be the $(g,u)$-th entry of $\hat{N}-N$. There exists a positive constant $C_1$ depending only on $A$ such that
	\begin{equation}\label{Nasy}
		\max_{1 \leq g,u \leq p}EM_{gu}^2 \leq C_1n^{-1}.
	\end{equation}
	
\end{lemma}
\begin{proof}
	Since N is diagonal,   when $g \neq u$,
	\begin{eqnarray*}
		M_{gu}={1 \over k}\sum_{h=1}^k\sum_{\tilde{u}=1}^p[\frac{1}{n}\sum_{i,j=1}^nf_h(s_i-s_j) \tilde Z_g(s_i) \tilde Z_{\tilde{u}}(s_j)][\frac{1}{n}\sum_{i,j=1}^nf_h(s_i-s_j) \tilde Z_u(s_i) \tilde Z_{\tilde{u}}(s_j)].
	\end{eqnarray*}
	
	Divide the term on the RHS of the above equation into three terms: (i)$\tilde{u}=g$,
	(ii)$\tilde{u}=u$ and (iii) $\tilde{u}\neq g,u$. We control each
	term as follows. When $\tilde{u}=g$,
	\begin{eqnarray*}
		E\Big([\frac{1}{n}\sum_{i,j=1}^nf_h(s_i-s_j) \tilde Z_g(s_i) \tilde Z_{g}(s_j)][\frac{1}{n}\sum_{i,j=1}^nf_h(s_i-s_j) \tilde Z_u(s_i) \tilde Z_{g}(s_j)]\Big)=0.
	\end{eqnarray*}
	\begin{eqnarray*}
		&&var\Big([\frac{1}{n}\sum_{i,j=1}^nf_h(s_i-s_j) \tilde Z_g(s_i) \tilde Z_{g}(s_j)][\frac{1}{n}\sum_{i,j=1}^nf_h(s_i-s_j) \tilde Z_u(s_i) \tilde Z_{g}(s_j)]\Big) \\
		=&&E\Big([\frac{1}{n}\sum_{i,j=1}^nf_h(s_i-s_j) \tilde Z_g(s_i) \tilde Z_{g}(s_j)][\frac{1}{n}\sum_{i,j=1}^nf_h(s_i-s_j) \tilde Z_u(s_i) \tilde Z_{g}(s_j)]\Big)^2\\
		=&&E\Big(n^{-4}\sum_{i_1,i_2,i_3,i_4,j_1,j_2,j_3,j_4=1}^nf_h(s_{i_1}-s_{j_1})f_h(s_{i_2}-s_{j_2})f_h(s_{i_3}-s_{j_3})f_h(s_{i_4}-s_{j_4})\\
		&& \tilde Z_g(s_{i_1}) \tilde Z_g(s_{i_1}) \tilde Z_g(s_{i_1}) \tilde Z_g(s_{i_1}) \tilde Z_{g}(s_{j_1}) \tilde Z_{g}(s_{j_3}) \tilde Z_{u}(s_{j_2}) \tilde Z_{u}(s_{j_4})\Big)\\
		\leq&&n^{-4}\sum_{i_1,i_2,i_3,i_4,j_1,j_2,j_3,j_4=1}^n[\prod_{v=1}^4\frac{A}{1+\|s_{i_v}-s_{j_v}\|^{d+\alpha}}]\frac{A}{1+\|s_{j_2}-s_{j_4}\|^{d+\alpha}}EZ^6_g(s) \\
		\leq&&  \tilde{C}_1n^{-1},
	\end{eqnarray*}
	where $\tilde{C}_1$ only depends on $A$. The first inequality is from (\ref{zcov})-(\ref{fcov}) and the independence between $Z_g(\cdot)$ and $Z_u(\cdot)$.  The second inequality is from (\ref{subgaussianassu}),  $C_0 \leq A$ and $\|s_i-s_j\|\geq \triangle$ for all $n \geq 2$ and $1 \leq i \neq j \leq n$.

	Thus we can control $$(\frac{1}{n}\sum_{i,j=1}^nf_h(s_i-s_j) \tilde Z_g(s_i) \tilde Z_{g}(s_j))(\frac{1}{n}\sum_{i,j=1}^nf_h(s_i-s_j) \tilde Z_u(s_i) \tilde Z_{g}(s_j)).$$
	When $\tilde{u}=u$, we can repeat the above method to control $$(\frac{1}{n}\sum_{i,j=1}^nf_h(s_i-s_j) \tilde Z_g(s_i) \tilde Z_{u}(s_j))(\frac{1}{n}\sum_{i,j=1}^nf_h(s_i-s_j) \tilde Z_u(s_i) \tilde Z_{u}(s_j)).$$
	Let's consider the third term $$\sum_{\tilde{u} \neq g,u}(\frac{1}{n}\sum_{i,j=1}^nf_h(s_i-s_j) \tilde Z_g(s_i) \tilde Z_{\tilde{u}}(s_j))(\frac{1}{n}\sum_{i,j=1}^nf_h(s_i-s_j) \tilde Z_u(s_i) \tilde Z_{\tilde{u}}(s_j)).$$
	We can rewrite it as
	\begin{eqnarray*}
		&&\frac{1}{n^2}\sum_{\tilde{u}\neq g,u}\sum_{i,j,\tilde{i},\tilde{j}=1}^nf_h(s_i-s_j)f_h(s_{\tilde{i}}-s_{\tilde{j}}) \tilde Z_g(s_i) \tilde Z_u(s_{\tilde{i}}) \tilde Z_{\tilde{u}}(s_j) \tilde Z_{\tilde{u}}(s_{\tilde{j}})\\
		=&&\frac{1}{n}\sum_{i,\tilde{i}=1}^n \Big(\frac{1}{n}\sum_{j,\tilde{j}=1}^nf_h(s_i-s_j)f_h(s_{\tilde{i}}-s_{\tilde{j}})\sum_{\tilde{u}\neq g,u} \tilde Z_{\tilde{u}}(s_j) \tilde Z_{\tilde{u}}(s_{\tilde{j}})\Big) \tilde Z_g(s_i) \tilde Z_u(s_{\tilde{i}}) .
	\end{eqnarray*}
	Let $\tilde H $ be a $n \times n$ symmetric matrix with $(i,\tilde{i})$th entry
	$$\frac{1}{n}\sum_{j,\tilde{j}=1}^nf_h(s_i-s_j)f_h(s_{\tilde{i}}-s_{\tilde{j}})\sum_{\tilde{u}\neq g,u} \tilde Z_{\tilde{u}}(s_j) \tilde Z_{\tilde{u}}(s_{\tilde{j}}).$$
	Recalling (\ref{zgtiledz}) and (\ref{Zasy0}), we define $Q=(I_n-n^{-1}1_{n \times n})\tilde H (I_n-n^{-1}1_{n \times n})$.
	Although $Q$ is random, we can find that $Q$ is independent of $Z_g(s)$ and $Z_u(s)$.
	It's easy to see
	\begin{eqnarray*}
		E\frac{1}{n}\sum_{i,j=1}^n Q_{i,j}  Z_g(s_i)Z_u(s_j)=0.
	\end{eqnarray*}
	\begin{eqnarray*}
		&&var[\frac{1}{n}\sum_{i,j=1}^n Q_{i,j}  Z_g(s_i)Z_u(s_j)]=E[\frac{1}{n}\sum_{i,j=1}^n Q_{i,j}  Z_g(s_i)Z_u(s_j)]^2\\
		=&& \frac{1}{n^2}\sum_{i,j,\tilde{i},\tilde{j}=1}^n  E(Q_{i,j}Q_{\tilde{i},\tilde{j}}) E[Z_g(s_i)Z_g(s_{\tilde{i}})] E[Z_u(s_j)Z_u(s_{\tilde{j}})]\\
		\leq &&  \frac{1}{n^2}\sum_{i,j,\tilde{i},\tilde{j}=1}^n  (EQ^2_{i,j})^{1/2}(EQ^2_{\tilde{i},\tilde{j}})^{1/2} \frac{A}{1+(s_i-s_{\tilde{i}})^{d+\alpha}}\frac{A}{1+(s_j-s_{\tilde{j}})^{d+\alpha}}\\
		\leq && \frac{\tilde{C}_2 }{n^2} \sum_{i,j=1}^nEQ_{i,j}^2=\frac{\tilde{C}_2 }{n^2}E \|Q\|_F^2,
	\end{eqnarray*}
	where $\tilde{C}_2$ only depends on $A$ and the first inequality is from (\ref{zcov}). The second inequality is from $\|s_i-s_j\|\geq \triangle$ for all $n \geq 2$ and $1 \leq i \neq j \leq n$. Recalling the definition of $Q$, we can rewrite it as
	\begin{eqnarray*}
		Q=\frac{1}{n}(I_n-n^{-1}1_{n \times n})V_h (I_n-n^{-1}1_{n \times n})   Z_{-g,-u}^\top Z_{-g,-u} (I_n-n^{-1}1_{n \times n})  V_h^\top(I_n-n^{-1}1_{n \times n}),
	\end{eqnarray*}
	where $V_h$ has the $(i,j)$th entry $f_h(s_i-s_j)$ and $Z_{-g,-u}$ is a $(p-2) \times n$ matrix without $Z^g$ and $Z^u$. Then
	
	\begin{eqnarray*}
		\|Q\|_F^2 \leq \|V_h\|^4\|\frac{1}{n}Z_{-g,-u}^\top Z_{-g,-u}\|_F^2 \leq \tilde{C}_3\|\frac{1}{n}Z^\top Z\|_F^2,
	\end{eqnarray*}
	where $\tilde{C}_3$ only depends on $A$ and the last inequality is from (\ref{fcov}).
	Moreover,
	\begin{eqnarray*}
		&&E\|\frac{1}{n}Z^\top Z\|_F^2=E\|\frac{1}{n}Z Z^\top \|_F^2\\
		=&&E\sum_{g,u=1}^p [n^{-1}\sum_{i=1}^nZ_g(s_i)Z_u(s_i)]^2\\
		=&&E\sum_{1 \leq g \neq u \leq p} [n^{-1}\sum_{i=1}^nZ_g(s_i)Z_u(s_i)]^2+E\sum_{g=1}^p [n^{-1}\sum_{i=1}^nZ^2_g(s_i)]^2\\
		=&&\sum_{1 \leq g \neq u \leq p}n^{-2}\sum_{i,j=1}^nE[Z_g(s_i)Z_g(s_j)]E[Z_u(s_i)Z_u(s_j)]+\sum_{g=1}^pn^{-2}\sum_{i,j=1}^nE[Z^2_g(s_i)Z^2_g(s_j)]\\
		\leq&&\sum_{1 \leq g \neq u \leq p}n^{-2}\sum_{i,j=1}^n(\frac{A}{1+\|s_{i}-s_{j}\|^{d+\alpha}})^2+\sum_{g=1}^pEZ^4_g(s)\\
		\leq&& \tilde{C}_4p ,
	\end{eqnarray*}
	where $\tilde{C}_4$ only depends on $A$. The first inequality is from (\ref{zcov}). The second equation is from (\ref{subgaussianassu}), $C_0 \leq A$, $p=o(n)$ and $\|s_i-s_j\|\geq \triangle$ for all $n \geq 2$ and $1 \leq i \neq j \leq n$.
	Then we can conclude that
	\begin{eqnarray*}
		E\|Q\|_F^2 \leq \tilde{C}_5p,
	\end{eqnarray*}
	where $\tilde{C}_5$ only depends on $A$. From $p=o(n)$,
	\begin{eqnarray*}
		var[\frac{1}{n}\sum_{i,j=1}^n Q_{i,j}  Z_g(s_i)Z_u(s_j)]\leq  \frac{\tilde{C}_2\tilde{C}_5 }{n^2}p=o(n^{-1}).
	\end{eqnarray*}
	Thus we control the third term and prove (\ref{Nasy}) for $g \neq u$.  When $g=u$, the proof is similar.
\end{proof}

\begin{definition}
	Let $J_1$ and $J_2$ be two subsets of $\{1,\cdots,p\}$. Let $\hat{N}_{J_1,J_2}$ be
	the  sub-matrix of $\hat{N}$ consisting of the rows with the indices in $J_1$ and
	the columns with the
	indices in $J_2$. Write $\hat{N}_{J_1}=\hat{N}_{J_1,J_1}$.
\end{definition}

\begin{lemma}\label{lemM}
	Under the conditions of Lemma \ref{lemCZ} and $J_1\cap J_2=\emptyset$,  we define the event $B_Z=\{n^{-1}\| Z\|^2 \leq  C_A\}$.
	Then  there exists a positive constant $C_2$ depending only on $A$, $c$ and $v$ such that
	\begin{equation}\label{Masy}
		P\Big(\|\hat{N}_{J_1,J_2}\|^2 > C_2n^{-1}v(|J_1|+|J_2|)\Big|B_Z\Big) \leq c(5^{|J_1|}+5^{|J_2|}) \exp(-5|J_1|v-5|J_2|v).
	\end{equation}
	Here $v>0$ can be finite or tending to infinite.
\end{lemma}
\begin{proof}
	Since $k$ is finite, it's sufficient to prove (\ref{Masy}) on
	$$n^{-2} Z_{J_1}(I_n-n^{-1}1_{n \times n}) V_h (I_n-n^{-1}1_{n \times n}) Z^\top  Z (I_n-n^{-1}1_{n \times n})   V_h^\top (I_n-n^{-1}1_{n \times n})   Z_{J_2}^\top, $$
	where $ Z_{J_1}$ is a sub-matrix of $ Z$ with $i$th row if and only if $i \in J_1$.  $ V_h$ is a $n \times n$ matrix with the $(i,j)$th entry $f_h(s_i-s_j)$. We define $\tilde V_h=(I_n-n^{-1}1_{n \times n}) V_h (I_n-n^{-1}1_{n \times n})$.
	\begin{eqnarray}\label{decomaa}
		\nonumber&& Z_{J_1}\tilde  V_h Z^\top  Z \tilde  V_h^\top  Z_{J_2}^\top\\
		\nonumber=&& Z_{J_1} \tilde  V_h Z_{J_1}^\top  Z_{J_1} \tilde   V_h^\top  Z_{J_2}^\top+ Z_{J_1} \tilde  V_h   Z_{J_2}^\top  Z_{J_2} \tilde   V_h^\top    Z_{J_2}^\top\\
		+&& Z_{J_1} \tilde  V_h Z_{J}^\top  Z_{J} \tilde   V_h^\top  Z_{J_2}^\top,
	\end{eqnarray}
	where $J$ is the complementary set of $J_1 \cup J_2$.
	At first we deal with $ Z_{J_1} \tilde  V_h Z_{J_1}^\top  Z_{J_1} \tilde   V_h^\top  Z_{J_2}^\top$.
	\begin{eqnarray*}
		&&\| Z_{J_1} \tilde  V_h Z_{J_1}^\top  Z_{J_1} \tilde  V_h^\top  Z_{J_2}^\top\|^2\\
		=&&\| Z_{J_2} \tilde  V_h Z_{J_1}^\top  Z_{J_1} \tilde V_h^\top Z_{J_1}^\top Z_{J_1} \tilde V_h Z_{J_1}^\top  Z_{J_1} \tilde  V_h^\top  Z_{J_2}^\top\|\\
		=&&\| Z_{J_2} H_{h,J_1} Z_{J_2}^\top\|,
	\end{eqnarray*}
	where $$ H_{h,J_1}=\tilde  V_h Z_{J_1}^\top  Z_{J_1}\tilde  V_h^\top Z_{J_1}^\top Z_{J_1}\tilde  V_h Z_{J_1}^\top  Z_{J_1} \tilde  V_h^\top$$
	is a $n \times n$ symmetric matrix with rank $|J_1|$ at most. Since $J_1\cap J_2=\emptyset$, $ H_{h,J_1}$ and $ Z_{J_2}^\top$ are independent.
	Moreover, under the event $B_Z=\{n^{-1}\| Z\|^2 \leq  C_A\}$,
	\begin{eqnarray*}
		\| H_{h,J_1}\|\leq \| \tilde V_h\|^4\| Z_{J_1}^\top Z_{J_1}\|^3\leq \| V_h\|^4\| Z^\top Z\|^3 \leq  \| V_h\|^4n^3C_A^3.
	\end{eqnarray*}
	It follows that
	\begin{eqnarray}\label{HhJ1bound}
		\lim_{n \rightarrow \infty}P(\| H_{h,J_1}\| \leq n^3\tilde{C}_A|B_Z )=1,
	\end{eqnarray}
	where $\tilde{C}_A$ only depends on $A$. Now we recall the rank of $ H_{h,J_1}$ is not larger than $|J_1|$. For given $ H_{h,J_1}$, we can do eigen-decomposition on it as follows.
	\begin{eqnarray}\label{HhJ1bound1}
		H_{h,J_1}= U_{h,J_1} \Lambda_{h,J_1} U_{h,J_1}^\top,
	\end{eqnarray}
	where $ U_{h,J_1}$ is a $n \times |J_1|$ matrix and $ \Lambda_{h,J_1}$ is a $|J_1| \times |J_1|$ diagonal matrix. $ U_{h,J_1}^\top U_{h,J_1}= I_{|J_1|}$. Then
	\begin{eqnarray*}
		\| Z_{J_1} \tilde  V_h Z_{J_1}^\top  Z_{J_1} \tilde   V_h^\top  Z_{J_2}^\top\|^2
		\leq \| Z_{J_2} U_{h,J_1}\|^2 \| \Lambda_{h,J_1}\|.
	\end{eqnarray*}
	Since $\| \Lambda_{h,J_1}\|$ can be controlled by (\ref{HhJ1bound}), we only need to consider $\| Z_{J_2} U_{h,J_1}\|^2$. Let $ Y= Z_{J_2} U_{h,J_1}$ be a $|J_2| \times |J_1|$ matrix with the $(i,j)$th entry $Y_{ij}$. The independence between the rows of $ Z_{J_2}$ implies the independence between the rows of $ Y$.

	For any fixed $1 \times |J_1|$ unit vector $ x=(x_1,\cdots,x_{|J_1|})$, we define $ x Y^\top$ as $ Y(x)=(y_{1}(x),\cdots,y_{|J_2|}(x))$. Then the elements of $ Y(x)$ are independent.
	\begin{eqnarray*}
		x Y^\top Y x^\top=\sum_{j=1}^{|J_2|} [y_{j}^2(x)-Ey_{j}^2(x)]+\sum_{j=1}^{|J_2|}Ey_{j}^2(x).
	\end{eqnarray*}
	$Yx^\top=Z_{J_2} U_{h,J_1}x^\top$ and $U_{h,J_1}x^\top$ is an unit vector independent of $Z_{J_2}$. By the sub-Gaussian property of $Z(s)$, we have
	\begin{eqnarray*}
		x Y^\top Y x^\top\leq \sum_{j=1}^{|J_2|} [y_{j}^2(x)-Ey_{j}^2(x)]+|J_2|\tilde{C}_{A,2},
	\end{eqnarray*}
	where $\tilde{C}_{A,2}$ only depends on $A$. Moreover, we can also deal with $\sum_{j=1}^{|J_2|} [y_{j}^2(x)-Ey_{j}^2(x)]$ with the sub-Gaussian property of $Z(s)$.
	Thus, for any fixed $1 \times |J_1|$ unit vector $ x$, any $c>0$ and $v>0$, there exists $C_{A,3}$ depending only on $A$, $c$ and $v$ such that

			\begin{eqnarray}\label{HhJ1bound2}
				P\Big(\| x Y^\top\|^2 > C_{A,3}v(|J_1|+|J_2|)\Big|B_Z\Big) \leq c\exp(-5|J_1|v-5|J_2|v).
			\end{eqnarray}
			As we know, the unit Euclidean sphere $S^{|J_1|-1}$ consists of all $|J_1|$-dimensional unit vectors $ x$.    Unfortunately,    the cardinality of $S^{|J_1|-1}$ are uncountable cardinal number. We can't use (\ref{HhJ1bound2}) to conclude the upper bound of $\| Y\|^2$ directly. Thus we use the method based on Nets to control $\| Y\|^2$. Let $ S_{\varepsilon}$ be a subset of $S^{|J_1|-1}$. For any $ x \in S^{|J_1|-1}$, there exists $ \tilde{x} \in  S_{\varepsilon}$ such that $\| \tilde{x}- x\| \leq \varepsilon$.
			Then if $\| Y x^\top\|=\| Y\|$, there exists $ \tilde{x} \in  S_{\varepsilon}$ such that
			\begin{eqnarray*}
				\| Y \tilde{x}^\top\| \geq \| Y x^\top\|- \| Y( \tilde{x}- x)^\top\|  \geq \| Y\|-\varepsilon\| Y\|=(1-\varepsilon)\| Y\|.
			\end{eqnarray*}
			Let $\varepsilon=1/2$,
			\begin{eqnarray*}
				\| Y\|^2 \leq 4 \max_{ \tilde{x} \in  S_{1/2}} \| Y \tilde{x}^\top\|^2.
			\end{eqnarray*}
			This, together with (\ref{HhJ1bound2}) and $| S_{\varepsilon}| \leq (1+2\varepsilon^{-1})^{|J_1|}$, implies that
			\begin{eqnarray}\label{HhJ1bound3}
				P\Big(\| Y\|^2 > 4C_{A,3}v(|J_1|+|J_2|)\Big|B_Z\Big)  \leq c5^{|J_1|} \exp(-5|J_1|v-5|J_2|v).
			\end{eqnarray}
			Recalling (\ref{HhJ1bound}), one can conclude that for any $c>0$, there exists $C_{A,4}$ only depending on $A$ and $c$ such that
			\begin{eqnarray}\label{HhJ1bound4}
				\nonumber&&P\Big(\| n^{-2}Z_{J_1} \tilde  V_h Z_{J_1}^\top  Z_{J_1} \tilde  V_h^\top  Z_{J_2}^\top\|^2 > 4C_{A,4}vn^{-1}(|J_1|+|J_2|)\Big|B_Z\Big) \\
				&&\leq c5^{|J_1|} \exp(-5|J_1|v-5|J_2|v).
			\end{eqnarray}
			Others term in  (\ref{decomaa}) can be controlled by the same method.
			This completes the proof.
		\end{proof}
		\begin{lemma}\label{lemNj}
			Under conditions A1-A3 and $p=o(n)$,
			\begin{equation}\label{Labmda2}
				\|\hat{N}_{J_i}- \Lambda_i\|=O_p(n^{-1/2}q_i^{1/2}),
			\end{equation}
			where $J_i=\{j \in \mathcal{Z}:p_{i-1} < j \leq p_{i}\}$,  $\Lambda_i = \diag(\la_{p_{i-1} +1}, \cdots, \la_{p_i})$, and $\la_i$ are specified
			in Condition A3.
			
		\end{lemma}
		\begin{proof}
			We divide $\hat{N}_{J_i}$ into two terms: (i) the diagonal term $\hat{N}_{J_i,d}$ and (ii) the off-diagonal term $\hat{N}_{J_i,o}$. Lemma \ref{t1} ensures $\|\hat{N}_{J_i,d}- \Lambda_i\|=O_p(n^{-1/2} q_i^{1/2})$. Thus we only need to show $\|\hat{N}_{J_i,o}\|=O_p(n^{-1/2}q_i^{1/2})$. If $q_i$ is finite, Lemma \ref{t1} can also ensure it. So we only need to consider the case $q_i$ tends to infinity.
			
			We can rewrite $\hat{N}_{J_i,o}$ and control $\|\hat{N}_{J_i,o}\|$ with the following idea.
			\begin{eqnarray*}
				\hat{N}_{J_i,o}=\begin{pmatrix}  V_{11} &   V_{12}\\
					V_{21} &   V_{22}
				\end{pmatrix}=\begin{pmatrix}  V_{11} &   0\\
					0 &   V_{22}
				\end{pmatrix}+\begin{pmatrix}  0 &   V_{12}\\
					V_{21} &   0
				\end{pmatrix}=D_1+V_{o,1}.
			\end{eqnarray*}
			Each block is a $q_i/2 \times q_i/2$ matrix.
			Note that $ V_{12}= V_{21}^\top$ and the norm of the second term $V_{o,1}$
			(off-diagonal block) can be controlled by $\| V_{12}\|$. Moreover, we can
			control $\| V_{12}\|$ by Lemmas \ref{lemCZ} and \ref{lemM}. In details, Lemma \ref{lemM} implies that
			\begin{equation}\label{Masyv11}
				P\Big(\| V_{o,1}\|^2 > C_2vn^{-1}q_i\Big|B_Z\Big) \leq c(5^{q_i/2}+5^{q_i/2}) \exp(-5q_iv).
			\end{equation}
			
			For the first term, we can
			repeat the step on $ V_{11}$ and $ V_{22}$ to get a new matrix with
			off-diagonal blocks as follows:
			\begin{eqnarray*}
				V_{o,2}=diag\Big[\begin{pmatrix}  0 &   V_{11,12}\\
					V_{11,21} &   0
				\end{pmatrix},\begin{pmatrix}  0 &   V_{22,12}\\
					V_{22,21} &   0
				\end{pmatrix}\Big].
			\end{eqnarray*}
			Lemma \ref{lemM} implies that
			\begin{equation}\label{Masyv12}
				P\Big(\| V_{o,2}\|^2 > C_2vn^{-1}q_i/2\Big|B_Z\Big) \leq 2c(5^{q_i/4}+5^{q_i/4}) \exp(-5q_iv/2).
			\end{equation}
			Repeat the steps, we can find that $V_{o,j}$ has $2^{j-1}$ diagonal  blocks and each diagonal block has two $2^{-j}q_i \times 2^{-j}q_i$ off-diagonal blocks. Lemma \ref{lemM} implies that
			\begin{equation}\label{Masyv1k}
				P\Big(\| V_{o,j}\|^2 > 2^{1-j}C_2vn^{-1}q_i\Big|B_Z\Big) \leq 2^{j-1}c(5^{2^{-j}q_i}+5^{2^{-j}q_i}) \exp(-5q_iv\times2^{1-j}).
			\end{equation}
			We  divide it into $j_0$ matrices: $ \hat{N}_{J_i,o}=\sum_{j=1}^{j_0}V_{o,j}$, $2^{j_0-1} \leq q_i$ and  $j_0=O(\log q_i)$.
			For different $j$, we choose different $v$ to control (\ref{Masyv1k}). When $\log q_i=o(2^{1-j}q_i)$, we choose $v=1$. It follows that
			\begin{equation}\label{Masyv1k1}
				P\Big(\| V_{o,j}\|^2 > 2^{1-j}C_2n^{-1}q_i\Big|B_Z\Big) \leq 2^{j-1}c(5^{2^{-j}q_i}+5^{2^{-j}q_i}) \exp(-5q_i\times2^{1-j})=o(\log^{-1} q_i).
			\end{equation}
			
			Otherwise, we choose $v=q_i^{4/5}\log^{-1} q_i$. It follows that
			\begin{eqnarray}\label{Masyv1k2}
				&&P\Big(\| V_{o,j}\|^2 > C_2n^{-1}q_i\log^{-2} q_i\Big|B_Z\Big)\\
				\nonumber\leq &&P\Big(\| V_{o,j}\|^2 > 2^{1-j}C_2q_i^{4/5}n^{-1}q_i\log^{-1} q_i\Big|B_Z\Big)\\
				\nonumber \leq&& 2^{j-1}c(5^{2^{-j}q_i}+5^{2^{-j}q_i}) \exp(-5q_i^{9/5}\log^{-1} q_i\times2^{1-j})=o(\log^{-1} q_i).
			\end{eqnarray}
			(\ref{Masyv1k1})-(\ref{Masyv1k2}) and $\|\hat{N}_{J_i,o}\| \leq \sum_{j=1}^{j_0} \| V_{o,j}\|$ imply that
			\begin{eqnarray}\label{Masyv1k3}
				&&P\Big(\|\hat{N}_{J_i,o}\| > 5C_2^{1/2}n^{-1/2}q_i^{1/2}\Big|B_Z\Big)=o(1).
			\end{eqnarray}
			Lemma \ref{lemCZ} implies that $\lim_{n\rightarrow \infty}P(B_Z)=1$. This, together with (\ref{Masyv1k3}) and $\|\hat{N}_{J_i,d}- \Lambda_i\|=O_p(n^{-1/2} q_i^{1/2})$, completes the proof.

		\end{proof}
		
		\begin{lemma}\label{prewrtienlem1}
			Under conditions A1-A2 and $p=o(n)$,
			\begin{equation}\label{prewrtieneq1}
				\|\Omega^\top\hat{\Sigma}^{-1}\Omega-I_p\|=O_p(n^{-1/2}p^{1/2}).
			\end{equation}
		\end{lemma}
		\begin{proof}
			Since $\tilde X(s_j)=\Omega \tilde Z(s_j)$,
			\begin{eqnarray*}
				\Omega^\top\hat{\Sigma}^{-1}\Omega-I_p=&&\Omega^\top[n^{-1} \sum_{1\le j \le n}\tilde  X(s_j) \tilde X(s_j)^\top]^{-1}\Omega-I_p\\
				=&&[n^{-1} \sum_{1\le j \le n}\tilde  Z(s_j) \tilde Z(s_j)^\top]^{-1}-I_p.
			\end{eqnarray*}
			It suffices to prove
			\begin{eqnarray*}
				\|n^{-1} \sum_{1\le j \le n} \tilde Z(s_j) \tilde Z(s_j)^\top-I_p\|=O_p(n^{-1/2}p^{1/2}).
			\end{eqnarray*}
			Following the proof of Lemma \ref{lemNj}, one can verify the above equation.
		\end{proof}
		
		\subsection{Proofs of Theorems}

		Recalling (\ref{nfihat}), write $\wh N = \wh \Gamma \wh \Lambda \wh \Gamma^{\top}$ as its spectral decomposition,
		i.e.
		\[
		\wh \Lambda = \diag(\wh \la_1, \cdots, \wh \la_p),
		\]
		where $\wh \la_1 \ge \cdots \ge \wh \la_p\ge 0$ are the eigenvalues of $\wh N$, and
		the columns of the orthogonal matrix $\wh \Gamma$ are the corresponding eigenvectors.
		Recalling the definition of $\wh W$ in (\ref{nfihatX})-(\ref{b5scaleh}), we can find that
		\begin{eqnarray*}
			\wh W=&&{1 \over k}\sum_{h=1}^k \hat M(f_h) \hat M(f_h)^\top\\
			=&&{1 \over k}\sum_{h=1}^k\big\{{1 \over n} \sum_{i,j=1}^n f_h(s_i - s_j)\wh \Sigma^{-1/2} \tilde   X(s_i)  \tilde X(s_j)^\top \wh \Sigma^{-1/2} \big\}\\
			&&\big\{{1 \over n} \sum_{i,j=1}^n f_h(s_i - s_j)\wh \Sigma^{-1/2} \tilde X(s_i)  \tilde  X(s_j)^\top \wh \Sigma^{-1/2} \big\}^\top\\
			=&&{1 \over k}\wh \Sigma^{-1/2}\Omega\sum_{h=1}^k\big\{{1 \over n} \sum_{i,j=1}^n f_h(s_i - s_j)  \tilde  Z(s_i)  \tilde Z(s_j)^\top \big\}\Omega^\top \wh \Sigma^{-1} \Omega\\
			&&\big\{{1 \over n} \sum_{i,j=1}^n f_h(s_i - s_j)  \tilde  Z(s_i) \tilde  Z(s_j)^\top \big\}^\top\Omega^\top\wh\Sigma^{-1/2}\\
			=&&{1 \over k}\wh \Sigma^{-1/2}\Omega\sum_{h=1}^k\big\{{1 \over n} \sum_{i,j=1}^n f_h(s_i - s_j)  \tilde   Z(s_i)  \tilde Z(s_j)^\top \big\}(\Omega^\top \wh \Sigma^{-1} \Omega-I_p)\\
			&&\big\{{1 \over n} \sum_{i,j=1}^n f_h(s_i - s_j)  \tilde  Z(s_i)  \tilde Z(s_j)^\top \big\}^\top\Omega^\top\wh\Sigma^{-1/2}+\wh \Sigma^{-1/2}\Omega\wh N\Omega^\top\wh\Sigma^{-1/2}.
		\end{eqnarray*}
		Let $\hat{\Sigma}^{-1/2}\Omega=\wh V_{\Omega} \wh \Lambda_{\Omega} \wh U_{\Omega}$ where $\wh V_{\Omega} \wh V_{\Omega}^\top=\wh U_{\Omega}\wh U_{\Omega}^\top=I_p$ and $\wh \Lambda_{\Omega}$ is a diagonal matrix. Then
		\begin{eqnarray*}
			\wh W=&&\wh V_{\Omega} \wh U_{\Omega} \wh \Gamma \wh \Lambda \wh \Gamma^{\top} \wh U_{\Omega}^\top\wh V_{\Omega}^\top+{1 \over k}\wh \Sigma^{-1/2}\Omega\sum_{h=1}^k\big\{{1 \over n} \sum_{i,j=1}^n f_h(s_i - s_j)  \tilde  Z(s_i)  \tilde Z(s_j)^\top \big\}\\
			&&\wh U_{\Omega}^\top (\wh \Lambda^2_{\Omega}-I_p) \wh U_{\Omega}\big\{{1 \over n} \sum_{i,j=1}^n f_h(s_i - s_j)  \tilde  Z(s_i) \tilde  Z(s_j)^\top \big\}^\top\Omega^\top\wh\Sigma^{-1/2}\\
			&&+\wh V_{\Omega} (\wh \Lambda_{\Omega}-I_p) \wh U_{\Omega} \wh \Gamma \wh \Lambda \wh \Gamma^{\top} \wh U_{\Omega}^\top\wh V_{\Omega}^\top+\wh V_{\Omega} \wh \Lambda_{\Omega} \wh U_{\Omega} \wh \Gamma \wh \Lambda \wh \Gamma^{\top} \wh U_{\Omega}^\top (\wh \Lambda_{\Omega}-I_p) \wh V_{\Omega}^\top.
		\end{eqnarray*}
		It follows that
		\begin{eqnarray*}
			\wh U_{\Omega}^\top\wh V_{\Omega}^\top \wh W\wh V_{\Omega} \wh U_{\Omega}=&& \wh \Gamma \wh \Lambda \wh \Gamma^{\top}+{1 \over k}\wh U_{\Omega}^\top \wh \Lambda_{\Omega} \wh U_{\Omega}\sum_{h=1}^k\big\{{1 \over n} \sum_{i,j=1}^n f_h(s_i - s_j)  \tilde  Z(s_i)  \tilde Z(s_j)^\top \big\}\\
			&&\wh U_{\Omega}^\top (\wh \Lambda^2_{\Omega}-I_p) \wh U_{\Omega}\big\{{1 \over n} \sum_{i,j=1}^n f_h(s_i - s_j)  \tilde  Z(s_i) \tilde  Z(s_j)^\top \big\}^\top\wh U_{\Omega}^\top \wh \Lambda_{\Omega} \wh U_{\Omega}\\
			&&+ \wh U_{\Omega}^\top (\wh \Lambda_{\Omega}-I_p) \wh U_{\Omega} \wh \Gamma \wh \Lambda \wh \Gamma^{\top} +\wh U_{\Omega}^\top \wh \Lambda_{\Omega} \wh U_{\Omega} \wh \Gamma \wh \Lambda \wh \Gamma^{\top} \wh U_{\Omega}^\top (\wh \Lambda_{\Omega}-I_p) \wh U_{\Omega}.
		\end{eqnarray*}
		Then
		\begin{eqnarray}\label{dfwN}
			\|\wh U_{\Omega}^\top\wh V_{\Omega}^\top\wh W\wh V_{\Omega} \wh U_{\Omega}- \wh \Gamma \wh \Lambda \wh \Gamma^{\top}\|=O\{\|\wh \Lambda_{\Omega}-I_p\|\|\wh \Lambda\|(1+\|\wh \Lambda_{\Omega}\|)^3\}.
		\end{eqnarray}
		(\ref{prewrtieneq1}) implies that $\|\wh \Lambda_{\Omega}-I_p\|=O_p(n^{-1/2}p^{1/2})$ and $\|\wh \Lambda_{\Omega}\|=O_p(1)$.
		
		Recalling $\hat{\Sigma}^{-1/2}\Omega=\wh V_{\Omega} \wh \Lambda_{\Omega} \wh U_{\Omega}$,
		\begin{eqnarray}\label{dfwn2}
			\|\wh U_W^{\top}  \wh \Sigma^{-1/2} \Omega-\wh U_W^{\top}  \wh V_{\Omega} \wh U_{\Omega}\| \leq \|\wh U_W^{\top} \wh V_{\Omega}^\top(\wh \Lambda_{\Omega}-I_p)\wh U_{\Omega}\|=O_p(n^{-1/2}p^{1/2}).
		\end{eqnarray}

		(\ref{dfwn2}) implies that the leading term of $\wh\Gamma_{\Omega}=\wh U_W^{\top} \wh \Sigma^{-1/2} \Omega$ is $\wh U_W^{\top}  \wh V_{\Omega} \wh U_{\Omega}$. (\ref{dfwN}) implies that
		$\wh U_W^{\top} \wh \Sigma^{-1/2} \Omega$ is close to $\wh \Gamma^{\top}$.

		Thus, the asymptotic properties of $\wh \Gamma^{\top}$ is the key point. We will prove the following theorem for $\wh \Gamma$ and $\wh \Lambda$.
		
		Put $q_i = p_i-p_{i-1}$ for $i=1, \cdots, m$ (see Condition A3), and
		\begin{equation}\label{eigendecGamma}
			\hat\Gamma=\begin{pmatrix}
				\hat\Gamma_{11} &  \cdots &
				\hat\Gamma_{1m}\\
				\cdots & \cdots & \cdots\\
				{\hat{\Gamma}}_{m1} &  \cdots & {\hat{\Gamma}}_{mm}
			\end{pmatrix},
			\qquad \hat\Lambda= \diag(\hat\Lambda_1,\cdots,\hat\Lambda_m),
		\end{equation}
		where submatrix $\wh \Gamma_{ij}$ is of the size $q_i \times q_j$, and $\hat\Lambda_i$
		is a $q_i \times q_i$ diagonal matrix.
		
		\begin{theorem}\label{t2forn}
			Let Conditions A1-A3 hold. As $n\to \infty$ and $p=o(n)$, it holds that
			\begin{equation}\label{eigendecGamma1}
				\|{\hat{\Gamma}}_{ij}\|=O_p\{n^{-1/2}(q_i+q_j)^{1/2}+n^{-1}p\}, \quad 1\le i \ne j \le m,
				\quad {\rm and}
			\end{equation}
			\begin{equation}\label{eigendecLambda1}
				\|{\hat{\Lambda}_i}-{\Lambda_i}\|=O_p(n^{-1/2}q_i^{1/2}+n^{-1}p), \quad 1 \le i \le m,
			\end{equation}
			where $\Lambda_i = \diag(\la_{p_{i-1} +1}, \cdots, \la_{p_i})$, and $\la_i$ are specified
			in Condition A3.
		\end{theorem}
		(\ref{dfwN}), (\ref{dfwn2}), (\ref{prewrtieneq1}) and Theorem \ref{t2forn} can conclude Theorem \ref{t2}. Thus, we now need to prove Theorem \ref{t2forn}.
		
		\begin{proof}[ Proof of Theorem \ref{t2forn}]
			(\ref{nfilamdba1b}) and (\ref{nfilamdba1max}) show that $m$ is bounded.
			Let $J_i=\{j \in \mathcal{Z}:p_{i-1} < j \leq p_{i}\}$.
			At first we prove (\ref{eigendecLambda1}). We only need to prove it when $i=1$ and other cases can be concluded by a permutation. Define $J_1^c$ be the complementary set of $J_1$, then we can rewrite $\det(\lambda I_p-\hat{N})=0$ as follows.
			\begin{eqnarray}\label{det0}
				0=\det(\lambda I_p-\hat{N})=\det\begin{pmatrix} \lambda I_{p_1}-\hat{N}_{J_1} &  -\hat{N}_{J_1,J_1^c}\\
					-\hat{N}_{J_1^c,J_1} &  \lambda I_{p-p_1}-\hat{N}_{J_1^c}
				\end{pmatrix}.
			\end{eqnarray}
			Lemmas \ref{lemCZ} and \ref{lemM} conclude $\|\hat{N}_{J_1^c,J_1}\|=O_p(n^{-1/2}p^{1/2})=o_p(1)$. Lemmas \ref{lemCZ}-\ref{lemNj} and the condition A3 imply  that there exists a positive constant $\tilde{C}_N$ such that
			\begin{eqnarray}\label{speei}
				\lim_{n \rightarrow \infty}P(\|\lambda_lI_{p-p_1}-\hat{N}_{J_1^c}\|_{min}>\tilde{C}_N)=1
			\end{eqnarray}
			for any $1 \leq l \leq p_1$.
			Lemma \ref{lemNj} also implies that
			\begin{eqnarray}\label{det2}
				\lim_{n\rightarrow \infty} P\Big(\lambda_{p_1}-\tilde{C}_N/2<\|\hat{N}_{J_1}\|_{min} \leq  \|\hat{N}_{J_1}\|<\lambda_{1}+\tilde{C}_N/2 \Big)=1.
			\end{eqnarray}
			
			If $\lambda \in (\lambda_{p_1}-\tilde{C}_N/2,\lambda_{1}+\tilde{C}_N/2)$ is a solution of (\ref{det0}), it is also (with probability 1) a solution of
			\begin{eqnarray}\label{det1}
				0=\det\Big(\lambda I_{p_1}-\hat{N}_{J_1}-\hat{N}_{J_1,J_1^c}(\lambda I_{p-p_1}-\hat{N}_{J_1^c})^{-1}\hat{N}_{J_1^c,J_1}\Big).
			\end{eqnarray}

			Lemma \ref{lemM} and  (\ref{speei}) imply that
			\begin{eqnarray}\label{det3}
				\|\hat{N}_{J_1,J_1^c}(\lambda I_{p-p_1}-\hat{N}_{J_1^c})^{-1}\hat{N}_{J_1^c,J_1}\|=O_p(n^{-1}p).
			\end{eqnarray}
			Let $\tilde{\lambda}_1 \geq \cdots \geq \tilde{\lambda}_{p_1}$ be the eigenvalues of $\hat{N}_{J_1}$, (\ref{det1})-(\ref{det3}) conclude that
			\begin{equation}\label{Labmda2a}
				\tilde{\lambda}_l-\hat{\lambda}_l=O_p(n^{-1}p)
			\end{equation}
			for any $1 \leq l \leq p_1$.
			This, together with (\ref{Labmda2}), concludes (\ref{eigendecLambda1}).
			
			Now we consider (\ref{eigendecGamma1}). We only need to prove it when $j=1$ and $i>1$. Other cases can be concluded by a permutation.
			From $\wh N = \wh \Gamma \wh \Lambda \wh \Gamma^{\top}$ and (\ref{eigendecGamma}), we can find that
			\begin{equation}\label{Gamma2}
				\begin{pmatrix} \sum_{i=1}^m\hat{N}_{J_1,J_i} \hat{\Gamma}_{i1} \\
					\cdots\\
					\sum_{i=1}^m\hat{N}_{J_m,J_i} \hat{\Gamma}_{i1}
				\end{pmatrix}=\hat{N}\begin{pmatrix}  \hat{\Gamma}_{11} \\
					\cdots\\
					\hat{\Gamma}_{m1}
				\end{pmatrix}=\begin{pmatrix}  \hat{\Gamma}_{11} \hat{\Lambda}_1 \\
					\cdots\\
					\hat{\Gamma}_{m1} \hat{\Lambda}_1
				\end{pmatrix}.
			\end{equation}
			Define $U_{11}=\hat{N}_{J_1,J_1}$, $U_{12}=\hat{N}_{J_1,J_1^c}$, $U_{21}=\hat{N}_{J_1^c,J_1}$ and $U_{22}=\hat{N}_{J_1^c,J_1^c}$. Similarly, define $\tilde{\Gamma}_{21}^\top=( \hat{\Gamma}_{21}^\top,\cdots, \hat{\Gamma}_{m1}^\top)^\top$. Then we can rewrite (\ref{Gamma2}) as
			\begin{equation}\label{Gamma3}
				\begin{pmatrix} U_{11} \hat{\Gamma}_{11}+U_{12}\tilde{\Gamma}_{21} \\
					U_{21} \hat{\Gamma}_{11}+U_{22}\tilde{\Gamma}_{21}
				\end{pmatrix}=\begin{pmatrix}  \hat{\Gamma}_{11} \hat{\Lambda}_1 \\
					\tilde{\Gamma}_{21}{\hat{\Lambda}}_1
				\end{pmatrix}.
			\end{equation}
			\begin{eqnarray*}
				\tilde{\Gamma}_{21} \hat{\Lambda}_1=\tilde{\Gamma}_{21}( \hat{\Lambda}_1-\lambda_1 I_{p_1})+\lambda_1\tilde{\Gamma}_{21}.
			\end{eqnarray*}
			Then the second line of (\ref{Gamma3}) is equivalent to
			\begin{eqnarray*}
				(U_{22}-\lambda_1 I_{p-p_1})\tilde{\Gamma}_{21}=\tilde{\Gamma}_{21}( \hat{\Lambda}_1-\lambda_1 I_{p_1})-U_{21} \hat{\Gamma}_{11}.
			\end{eqnarray*}
			Recalling (\ref{speei}), $U_{22}-\lambda_1 I_{p-p_1}$ is invertible with probability 1 as $n$ tends to infinity.
			\begin{eqnarray*}
				\tilde{\Gamma}_{21}=(U_{22}-\lambda_1 I_{p-p_1})^{-1}\tilde{\Gamma}_{21}( \hat{\Lambda}_1-\lambda_1 I_{p_1})-(U_{22}-\lambda_1 I_{p-p_1})^{-1}U_{21} \hat{\Gamma}_{11}.
			\end{eqnarray*}
			(\ref{nfilamdba1a})-(\ref{nfilamdba1b}) and Lemmas \ref{lemCZ}-\ref{lemNj} imply that $\| \hat{\Lambda}_1-\lambda_1 I_{p_1}\|=o_p(1)$ and $\|(U_{22}-\lambda_1 I_{p-p_1})^{-1}\|=O_p(1)$. Then $(\lambda_1 I_{p-p_1}-U_{22})^{-1}U_{21}\hat{\Gamma}_{11}$ is the leading term of $\tilde{\Gamma}_{21}$. Moreover, $\|\hat{\Gamma}_{11}\|=O(1)$.
			Thus we only need to consider $(\lambda_1 I_{p-p_1}-U_{22})^{-1}U_{21}$. We rewrite $(\lambda_1 I_{p-p_1}-U_{22})^{-1}$ as
			\begin{equation*}
				\begin{pmatrix} \lambda_1I_{p_2}-\hat{N}_{J_2,J_2}&\cdots&-\hat{N}_{J_2,J_m} \\
					\cdots&\cdots&\cdots\\
					-\hat{N}_{J_m,J_2}&\cdots&\lambda_1I_{p_m}-\hat{N}_{J_m,J_m}
				\end{pmatrix}^{-1}=(\lambda_1 I_{p-p_1}-U_{22})^{-1}=\begin{pmatrix} V_{22}&\cdots&V_{2m} \\
					\cdots&\cdots&\cdots\\
					V_{m2}&\cdots&V_{mm}
				\end{pmatrix}.
			\end{equation*}
			(\ref{nfilamdba1a})-(\ref{nfilamdba1b}) and Lemma \ref{lemNj} ensure $\|(\lambda_1I_{p_i}-\hat{N}_{J_i,J_i})^{-1}\|=O_p(1)$ for $2 \leq i \leq m$. Lemma \ref{lemM} ensures $\|\hat{N}_{J_i,J_t}\|=O_p(n^{-1/2}p^{1/2})=o_p(1)$ for $2 \leq i \neq t \leq m$. Since $m$ is finite,  we can find  $\|V_{ii}\|=O_p(1)$ and
			$\|V_{it}\|=O_p(n^{-1/2}p^{1/2})$ for $2 \leq i \neq t \leq m$. Recall that $\|\hat{N}_{J_i,J_1}\|=O_p(n^{-1/2}(q_1+q_i)^{1/2})$ for $2 \leq i  \leq m$ and
			\begin{equation*}
				(\lambda_1 I_{p-p_1}-U_{22})^{-1}U_{21}=\begin{pmatrix} V_{22}&\cdots&V_{2m} \\
					\cdots&\cdots&\cdots\\
					V_{m2}&\cdots&V_{mm}
				\end{pmatrix}\begin{pmatrix} \hat{N}_{J_2,J_1} \\
					\cdots\\
					\hat{N}_{J_m,J_1}
				\end{pmatrix}.
			\end{equation*}
			It follows that
			$\|V_{ii}\hat{N}_{J_i,J_1}\|=O_p(n^{-1/2}(q_1+q_i)^{1/2})$ and $\|\sum_{t\neq i}V_{it}\hat{N}_{J_t,J_1}\|=O_p(n^{-1}p)$.
			We complete the proof of (\ref{eigendecGamma1}).

			\end{proof}
			Now we prove Theorem \ref{t3}. By the same idea, we give the following result for $\wh N$.
			\begin{theorem} \label{t3forN}
				Let conditions A1, A2 and A4 hold. Denote by $\hat \gamma_{ij}$ the $(i,j)$-th
				entry of matrix $\wh \Gamma $ in (\ref{eigendecGamma}). Then as $n, p \to
				\infty$, it holds that

				\begin{equation}\label{eigendecGamma2gap}
					\wh \gamma_{ij}=O_p(n^{-1/2}v^{-1}_{\rm gap}|j-i|^{-1}) \quad {\rm for} \;\;
					1\le i\ne j \le p,
					\quad {\rm and}
				\end{equation}
				\begin{equation}\label{eigendecLambda2gape}
					\wh \gamma_{ii}=1+O_p(n^{-1}v^{-2}_{\rm gap}) \quad {\rm for} \;\; i=1, \cdots, p.
				\end{equation}
				Moreover,
				\begin{equation}\label{eigendecLambda1aa}
					\|{\hat{\Lambda}}-{\Lambda}\|=O_p(n^{-1/2}p^{1/2}).
				\end{equation}

			\end{theorem}

			\begin{proof}[ Proof of Theorem \ref{t3forN}]
				Following the proof of Lemma \ref{lemNj}, one can verify that $\|\wh \Lambda-N\|=O_p(n^{-1/2}p^{1/2})$. This, together with A4, implies (\ref{eigendecLambda1aa}).

					From $\wh N \wh \Gamma= \wh \Gamma \wh \Lambda $, we can find that
					\begin{eqnarray}\label{reeigendeco}
						\wh \Gamma \wh \Lambda-N \wh \Gamma=(\wh N-N) \wh \Gamma.
					\end{eqnarray}
					(\ref{reeigendeco}) implies that
					\begin{eqnarray}\label{reeigendeco2a}
						\wh \gamma_{ij}(\wh \lambda_{j}-\lambda_{i})=\sum_{s=1}^p M_{is} \wh \gamma_{sj},
					\end{eqnarray}
					where $M_{is}$ is defined in Lemma \ref{t1}.
					The condition A4 and $\|\wh \Lambda-N\|=O_p(n^{-1/2}p^{1/2})$ can control $(\wh \lambda_{j}-\lambda_{i})$.  Then we can divide the right hand of the above equation into two part.
					\begin{eqnarray}\label{reeigendeco2}
						\sum_{s=1}^{p} M_{is} \wh \gamma_{sj}=\sum_{s \neq j} M_{is} \wh \gamma_{sj}+M_{ij} \wh \gamma_{jj}.
					\end{eqnarray}
					(\ref{Nasy}) implies that $E|M_{ij} \wh \gamma_{jj}|^2 \leq E|M_{ij}|^2 \leq C_1n^{-1}$. Thus we only need to consider the order of $\sum_{s \neq j} M_{is} \wh \gamma_{sj}$. Define $v=\max_{1 \leq i \leq p}\max_{j \neq i}|\sum_{s \neq j} M_{is} \wh \gamma_{sj}|$. Then for any $j \neq i$, (\ref{reeigendeco2a}) implies that $$|\wh \gamma_{ij}| \leq (|i-j|v_{gap}-\|\wh \Lambda-N\|)^{-1}(v+|M_{ij}|)$$
					and
					\begin{eqnarray*}
						&&|\sum_{s \neq j} M_{is} \wh \gamma_{sj}|\leq  \sum_{s \neq j} |M_{is}| | \wh \gamma_{sj}|  \\
						\leq&& \sum_{s \neq j} |M_{is}| (|s-j|v_{gap}-\|\wh \Lambda-N\|)^{-1}(v+|M_{sj}|)\\
						\leq&&v \sum_{s \neq j} |M_{is}|(|s-j|v_{gap}-\|\wh \Lambda-N\|)^{-1}+ \sum_{s \neq j} |M_{is}||M_{sj}|(|s-j|v_{gap}-\|\wh \Lambda-N\|)^{-1}.
					\end{eqnarray*}
					The condition A4, $\|\wh \Lambda-N\|=O_p(n^{-1/2}p^{1/2})$ and (\ref{Nasy}) conclude that
					$$\sum_{s \neq j} |M_{is}|(|s-j|v_{gap}-\|\wh \Lambda-N\|)^{-1}=O(v_{gap}^{-1}\log p\max_{1 \leq i,s \leq p}|M_{is}|)=o_p(1)$$ and
					$$\sum_{s \neq j} |M_{is}||M_{sj}|(|s-j|v_{gap}-\|\wh \Lambda-N\|)^{-1}=o_p(n^{-1/2}).$$
					This, together with the definition of $v$, implies that $v=o_p(n^{-1/2})$. $$|\wh \gamma_{ij}| \leq (|i-j|v_{gap}-\|\wh \Lambda-N\|)^{-1}[o_p(n^{-1/2})+|M_{ij}|].$$
					This, together with (\ref{Nasy}), concludes  (\ref{eigendecGamma2gap}).
					$$\wh \gamma_{ii}^2=1-\sum_{j \neq i} \wh \gamma_{ij}^2 \geq 1-\sum_{j \neq i} (|i-j|v_{gap}-\|\wh \Lambda-N\|)^{-2}(v+|M_{ij}|)^2=1+O_p(n^{-1}v_{gap}^{-2}) .$$
					We complete the proof.
				\end{proof}
						
				(\ref{dfwN})  and (\ref{prewrtieneq1}) imply that
				\begin{eqnarray*}
					\|\wh \Gamma^{\top}\wh U_{\Omega}^\top\wh V_{\Omega}^\top \wh W\wh V_{\Omega} \wh U_{\Omega}\wh \Gamma- \wh \Lambda\|=O_p(n^{-1/2}p^{1/2}).
				\end{eqnarray*}

				This and Theorem \ref{t3forN} can conclude the asymptotic properties of $\wh U_W^{\top}  \wh V_{\Omega} \wh U_{\Omega}\wh \Gamma$.
				Then we can prove Theorem \ref{t3} by (\ref{dfwn2}) and Theorem \ref{t3forN}.

			\end{document}

%% file: QYalias.tex
 \@addtoreset{equation}{section}
 
\newtheorem{definition}{Definition}[section]
\newtheorem{theorem}{Theorem}
\newtheorem{lemma}{Lemma} 
\newtheorem{corollary}{Corollary}
\newtheorem{proposition}{Proposition}

\renewcommand{\hat}{\widehat}
\def\singlespace{\def\baselinestretch{1}\@normalsize}

\def\wh{\widehat}

\newcommand{\cov}{{\rm Cov}}
\newcommand{\diag}{{\rm diag}}

\newcommand{\var}{{\rm Var}}

\def\la{\lambda}

\newcommand{\pcon}{\stackrel{P}{\longrightarrow}}

\newcommand{\calS}{{\mathcal S}}

\newcommand{\RR}{\EuScript R}

\def\6bullets{\bullet\bullet\bullet\bullet\bullet\bullet}

\DeclareMathAlphabet\EuScriptBF{U}{eus}{b}{n}

\def\AS{{\sl The Annals of Statistics}}
\def\JMA{{\sl Journal of multivariate analysis}}